\documentclass[journal]{vgtc}                

\ifpdf                               
\pdfoutput=1\relax                   
\usepackage{graphicx}                
\DeclareGraphicsExtensions{.pdf,.png,.jpg,.jpeg} 
\else                                
\usepackage{graphicx}                
\DeclareGraphicsExtensions{.eps}     
\fi

\graphicspath{{figures/}{pictures/}{images/}{./}} 

\usepackage{microtype}                 
\PassOptionsToPackage{warn}{textcomp}  
\usepackage{textcomp}                  
\usepackage{mathptmx}                  
\usepackage{times}                     
\usepackage{cite}                      
\usepackage{tabu}                      
\usepackage{booktabs}                  

\usepackage{wrapfig}
\usepackage{todonotes}
\usepackage{mathptmx}
\usepackage{graphicx}
\usepackage{times}
\usepackage{amsmath,amscd,amssymb, amsthm}
\usepackage{times}
\usepackage[abs]{overpic}
\usepackage{cases}
\usepackage{float}
\usepackage{caption}
\usepackage{subfloat}
\usepackage{subfig}
\usepackage{wrapfig}
\usepackage{placeins}
\usepackage{multirow}
\usepackage{paralist}

\onlineid{1360}

\newcommand{\sref}[1]{Section~\ref{#1}}

\newcommand{\eg}{e.g.,~}
\newcommand{\ie}{i.e.,~}

\long\def\symbolfootnote[#1]#2{\begingroup%
	\def\thefootnote{\fnsymbol{footnote}}\footnote[#1]{#2}\endgroup}

\newtheorem{theorem}{\sffamily Theorem}
\newtheorem{lemma}[theorem]{\sffamily Lemma}
\newtheorem{corollary}[theorem]{\sffamily Corollary}

\DeclareMathOperator{\trace}{trace}

\vgtccategory{Research}
\vgtcpapertype{algorithm/technique}

\title{Global Topology of 3D Symmetric Tensor Fields}

\author{
Shih-Hsuan Hung,
Yue Zhang, \textit{Member,~IEEE},
Eugene Zhang, \textit{Senior Member,~IEEE}
}

\authorfooter{
\item
Shih-Hsuan Hung is with the School of Electrical Engineering and Computer Science, Oregon State University. E-mail: hungsh@oregonstate.edu.
\item
Yue Zhang is an Associate Professor with the School of Electrical Engineering and Computer Science, Oregon State University. E-mail: zhangyue@oregonstate.edu.
\item
Eugene Zhang is a Professor with the School of Electrical Engineering and Computer Science, Oregon State University. E-mail: zhange@eecs.oregonstate.edu.
}

\shortauthortitle{Hung \MakeLowercase{\textit{et al.}}: Global Topological Analysis of 3D Symmetric Tensor Fields}

\abstract{
There have been recent advances in the analysis and visualization of 3D symmetric tensor fields, with a focus on the robust extraction of tensor field topology. However, topological features such as degenerate curves and neutral surfaces do not live in isolation. Instead, they intriguingly interact with each other. In this paper, we introduce the notion of {\em topological graph} for 3D symmetric tensor fields to facilitate global topological analysis of such fields. The nodes of the graph include degenerate curves and regions bounded by neutral surfaces in the domain. The edges in the graph denote the adjacency information between the regions and degenerate curves. In addition, we observe that a degenerate curve can be a loop and even a knot and that two degenerate curves (whether in the same region or not) can form a link. We provide a definition and theoretical analysis of individual degenerate curves in order to help understand why knots and links may occur. Moreover, we differentiate between wedges and trisectors, thus making the analysis more detailed about degenerate curves. We incorporate this information into the topological graph. Such a graph can not only reveal the global structure in a 3D symmetric tensor field but also allow two symmetric tensor fields to be compared. We demonstrate our approach by applying it to solid mechanics and material science data sets.
} 


\CCScatlist{ 
\CCScat{K.6.1}{Management of Computing and Information Systems}%
{Project and People Management}{Life Cycle};
\CCScat{K.7.m}{The Computing Profession}{Miscellaneous}{Ethics}
}

\teaser{
\centering
\begin{overpic}[width={\textwidth}]{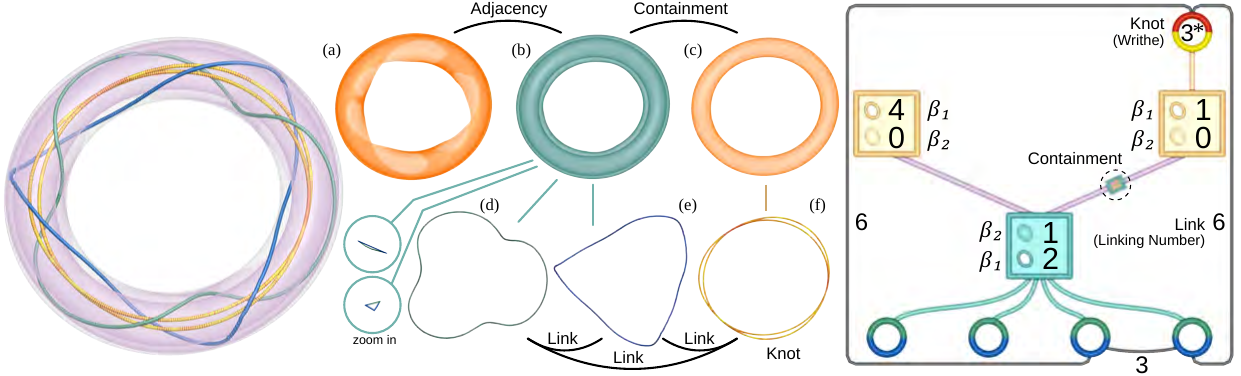}
\end{overpic}
\caption{In a simulated stress tensor field inside an O-ring, a linear region ((b): green) is sandwiched between two planar regions ((a) and (c): orange), which reflects the fact that compressive stress has been applied to the boundary of the O-ring. The innermost planar region (c) contains a single degenerate loop of the wedge type that travels inside the O-ring twice ((f): the yellow loop), indicating the existence of an equilibrium of uniaxial compression of material at the core of the O-ring. On the other hand, the linear region (b) contains a trisector loop ((e): the blue loop), where material is being pushed away in several directions, and a wedge loop ((d): the green loop), where material extension is constrained by the boundary of the O-ring. In addition, the two degenerate loops form a link due to the non-uniformity and periodicity in the stress at the boundary. The topological features (regions and degenerate curves) along with the interactions among them are encoded in the topological graph (right) that we introduce in this paper.
}
\label{fig:teaser}
}

\renewcommand{\manuscriptnotetxt}{}

\vgtcinsertpkg

\begin{document}


\firstsection{Introduction}\label{sec:introduction}
\maketitle
Tensor fields are widely used in solid mechanics and material science. In these domains, the topological features of the stress tensor field have explicit physical meanings. For example, degenerate curves represent uniaxial extension and compression while neutral surfaces represent pure shear~\cite{CRISCIONE:00}.

There have been some recent advances in the topological analysis and visualization of 3D symmetric tensor fields~\cite{Zheng:04,Zheng:05a,Palacios:16,Roy:19,Qu:21}, which not only introduce the topology of such fields but also provide robust algorithms to extract individual topological features.

However, topological features in a tensor field do not live in isolation. Instead, there are intricate relationships among degenerate curves and neutral surfaces. For example, a degenerate curve can form a loop and even a knot. Two degenerate curves can form a link (see Figure~\ref{fig:teaser} for examples). In addition, neutral surfaces can divide the domain into regions, inside each of which the stress tensor field has a uniform behavior of either {\em extension-dominant} (linearity) or {\em compression-dominant} (planarity). A region can contain degenerate curves and other regions in its interior, thus having a complicated topological structure (Figure~\ref{fig:teaser}).

\begin{figure}[!t]
	\centering
	\begin{overpic}[width={\columnwidth}]{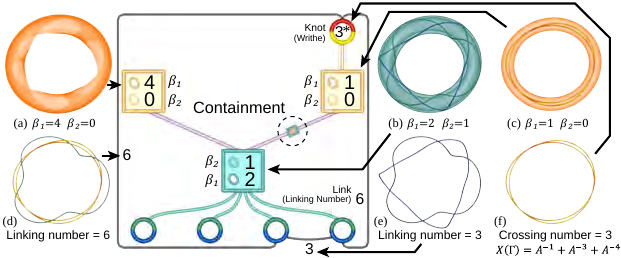}
	\end{overpic}
	\caption{The components of our topological graph. Cyan and yellow squares represent linear and planar regions, respectively, with their Betti numbers written inside. A pair of adjacent planar and linear regions are connected by an edge in the graph. The rectangular glyph on an edge indicates that one of the regions in the pair is inside the other region. Degenerate curves are the second type of nodes in the graph, represented as circular rings (open for open curves and closed for closed curves). Their Writhe numbers are computed and visualized. The linking numbers between linked pairs of degenerate curves are also included in the visualization.
	}
	\label{fig:graph}
\end{figure}

In this paper, we introduce a topological graph for 3D symmetric tensor fields (Figure~\ref{fig:teaser}: right). The nodes of the graph consist of degenerate curves as well as uniformly linear regions (extension-dominant) and uniformly planar regions (compression-dominant). An edge in the graph can indicate a pair of adjacent regions, a region containing a degenerate curve, or a pair of linked degenerate curves. In addition, unlike scalar fields whose topological features consist of isolated points, in 3D tensor fields the topology includes both curves (degenerate curves) and volumes (regions). Thus, a topological feature can have a non-trivial topology, such as a knotted degenerate curve and a region that contains multiple air bubbles (other regions). Consequently, we compute a number of characteristics such as the {\em knottiness} of a degenerate curve ({\em Writhe number}) and the homology of a region (Betti numbers). Similarly, an edge in the graph can also indicate a rather complicated relationship between two nodes, such as the linking of two degenerate curves. We compute the linking number for such an edge. Figure~\ref{fig:graph} provides a more detailed annotation of this graph. The definition of Betti numbers is given in Section~\ref{sec:topological_graphs}, and the definitions of the Writhe number, the linking number, and the Jones polynomials are given in Section~\ref{sec:graph_construction}.

Our topological graph provides a holistic view of the topological features in the field, leading to insight that is difficult to obtain by analyzing individual features in isolation. As an example, Figure~\ref{fig:teaser} shows the topological features of the stress tensor field inside an O-ring simulation data using existing topology-driven visualization approaches (Figure~\ref{fig:teaser} (left)). While it is clear from the visualization that there are degenerate curves and neutral surfaces, it is difficult to see why the degenerate curves appear in these locations. The topology of the neutral surface is also hard to discern from the visualization. With our topological graph (Figure~\ref{fig:teaser}: right), we see a linear region (cyan square) sandwiched by two planar regions (yellow squares). By selecting their nodes in the graph, we can see the actual regions (Figure~\ref{fig:teaser} (a)-(c)). It can be observed that, as the boundary force pushes the material towards the core of the O-ring, there is a linear region ((b): green region) that indicates extension. This extension in the material leads to strong compression at the core of the O-ring ((c): the yellow region), where the compressed material has nowhere to go. Consequently, an equilibrium of uniaxial compression appears in the form of a degenerate curve ((f): yellow). Due to the presence of a compression load everywhere on the boundary, the equilibrium cannot reach the boundary, thus forcing the degenerate curve into a loop. In addition, the three-fold rotational symmetry in the boundary loading condition, a global property, is captured by the shapes of the two degenerate curves ((d) and (e)) in the linear region (b: region) as well as the six crossing points between the two {\em linked} curves. Such global analysis is difficult to obtain by only considering individual degenerate curves and neutral surfaces.

\begin{figure}[!t]
	\centering
	\begin{overpic}[width={\columnwidth}]{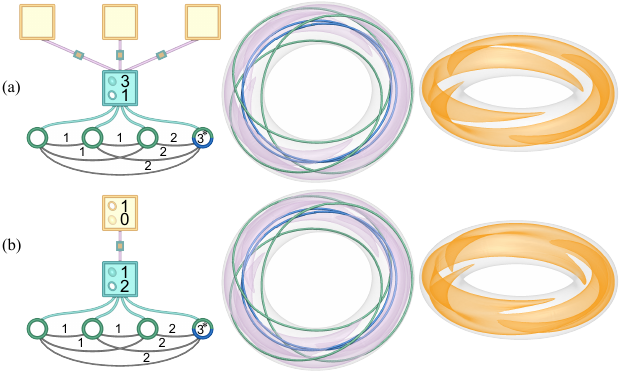}
	\end{overpic}
	\caption{This figure compares two stress tensor fields where the boundary load for (b) is $7\%$ higher than that for (a). The right column shows the boundary of the compression-dominant regions (yellow squares in the graphs).
	}
	\label{fig:global_l}
\end{figure}

Our topological graph also allows two tensor fields to be compared, such as the stress tensor fields in the O-ring given  two different boundary loading conditions (Figure~\ref{fig:global_l}). The middle column shows the degenerate curves and neutral surfaces extracted using existing techniques~\cite{Roy:19,Qu:21}. Notice that the visualizations of the two fields look similar. However, their topological graphs show both similarities and differences. On the one hand, both fields contain a linear region (green squares in the graphs) which contains four degenerate loops. Furthermore, the four degenerate curves are linked with the same linking numbers, and one of the loops is a trefoil knot. This indicates that, despite the difference in the magnitude of the loading conditions for the two scenarios, the stress at the core of the O-ring does not differ significantly. On the other hand, the field in (a) has three compression-dominant regions, all of which are contractible. In contrast, the field in (b), which has a larger load, has the three regions merged into one and is no longer contractible. This indicates that compression area on the boundary is much larger due to the increased load.

With topological graphs, global analysis of tensor fields is enabled which has the potential of shedding new insight to domain scientists than showing only the visualizations (Figure~\ref{fig:teaser} (left) and Figure~\ref{fig:global_l} (middle)).

Another important aspect of this paper is the differentiation of wedge-type degenerate points from trisector-type degenerate points in 3D symmetric tensor fields. Since their introduction in the early $2000$s by Zheng et al.~\cite{Zheng:05b}, there has been relatively little additional research on the topic and application to engineering data. Most research focuses on the linearity/planarity aspect of degenerate points. However, such an approach can lead to incomplete or misleading interpretation of the data. For example, as shown Figure~\ref{fig:cmp} (a), the linear trisector loop (the blue loop) is close to the maximal compression force on the boundary, where material is pushed away, and the wedge loop (the green loop) is close to the minimal compression force on the boundary, where the material flows to but is stopped by the O-ring (thus a dead-end). In Figure~\ref{fig:cmp} (b), the wedge/trisector classification is unaccounted for, and it is no longer clear why two linear degenerate loops (green) appear there and how they relate to each other.

\begin{figure}[!t]
	\centering
	\begin{overpic}[width={0.8\columnwidth}]{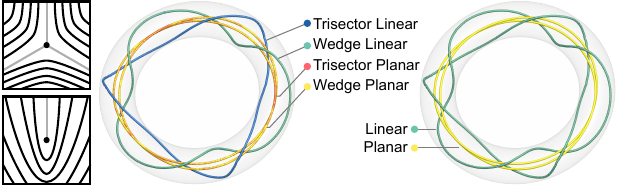}
		\put( 60, -7) {\small(a)}
		\put(160, -7) {\small(b)}
	\end{overpic}
	\caption{The linear wedge curve ((a): green) and the linear trisector curve ((a): blue) allow us to identify the places where the materials are being pushed out (near the trisector curve) and where the materials are being pushed into a dead-end (near the wedge curve). Not differentiating trisectors from wedges (b) makes this insight unavailable and can lead to incomplete interpretation.
	}
	\label{fig:cmp}
\end{figure}

\begin{figure}[!t]
	\centering
	\begin{overpic}[width={0.9\columnwidth}]{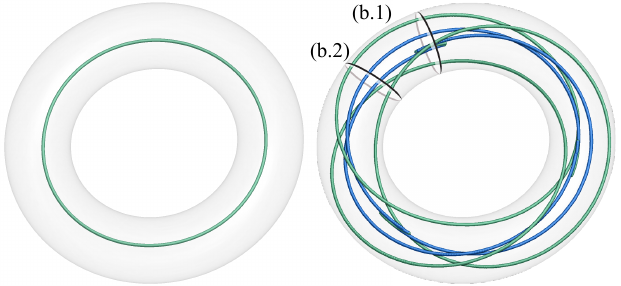}
		\put( 52, -7) {\small(a)}
		\put(165, -7) {\small(b)}
	\end{overpic}
	\caption{
		In an O-ring simulation scenario where the magnitude of the load is constant on the boundary (a), there is a single degenerate loop at the core of the O-ring. When the magnitude of the load at the boundary becomes anisotropic, additional degenerate curves results. Yet, on each cut plane (e.g., b.1 and b.2), the difference between the number of wedge points and the number of trisector points is $1$, which is the same as the field in (a).
	}
	\label{fig:global_3}
\end{figure}

Enabled by the wedge/trisector differentiation, we consider two O-ring simulation scenarios (Figure~\ref{fig:global_3}). In (a), the load has a constant magnitude on the boundary, which leads to a single degenerate loop (green) that is the equilibrium in the stress tensor field. When the load's magnitude becomes anisotropic but periodic (b), additional degenerate loops appear that include both wedge curves (green) and trisector curves (blue). It is interesting to observe that on each cross section of the O-ring, there are either $4$ wedges and $3$ trisectors (b.1), or $3$ wedges and $2$ trisectors (b.2). The difference between the number of wedges and the number of trisectors is one. This seems to suggest that the total tensor index on the intersection between each cut plane and the O-ring is constant in this case. We hypothesize that tensor field topology is constrained by the boundary loading condition as well as the topology of the O-ring, similar to the Poincar\'e-Hopf theorem that states the total singularity index of a vector field is determined by the Euler characteristic of the domain. Such insight cannot be achieved by only considering individual topological features and without differentiating between wedges and trisectors.

To evaluate the effectiveness of the topological graph as well as the additional characterization of wedge/trisector degenerate points, we apply these ideas to a number of simulation data sets with applications in solid mechanics and material science. To summarize, in this paper we make the following contributions:

\begin{itemize}
\item We introduce the notion of {\em topological graphs} for 3D symmetric tensor fields.
\item As part of the graph we introduce new topological features such as linear regions and planar regions, which are nodes in the graph. We also propose to measure the complexity of these regions by computing their homology in the form of Betti numbers.
\item We observe that degenerate curves can be knotted and be linked to each other. This information is captured by computing the curve characteristics such as the Writhe numbers and Jones polynomials.
\item We advocate the addition of the wedge/trisector classification of degenerate curves to physical interpretations. Furthermore, we introduce the index of degenerate curves that can be used to differentiate between wedges and trisectors.
\end{itemize}

\section{Related Work}
\label{sec:related_work}

Tensor field visualization is an important area of research that has seen waves of advances over the past decades~\cite{EPFL-BOOK-138668,Kratz:13}. Topology-driven analysis and visualization of tensor fields have found many applications in understanding solid and fluid mechanics data as well as material science. Inspired by the use of vector field topology in fluid dynamics, Delmarcelle and Hesselink~\cite{delmarcelle:visualizing} introduce the notion of tensor field topology of 2D symmetric tensor fields in terms of degenerate tensors with repeating eigenvalues. To further understand the topological features, Leeuw and van Liere~\cite{de:1999:collapsing} propose a topological graph with the relationship of the degenerate points. Jankowai et al.~\cite{jankowai:2019:robust} introduce a diagram of the degenerate points with a tree structure based on the robustness of a tensor field near the degenerate points.

The concept of the tensor field topology of 2D symmetric tensor fields is extended to 3D symmetric tensor fields by Hesselink et al.~\cite{hesselink:topology}, with a focus on understanding the behaviors of tensor fields near triple degenerate points, where the tensor field has a value of a multiple of the identity tensor. Later, using dimensionality analysis, Zheng and Pang~\cite{Zheng:04} point out that triple degenerate points are structurally unstable as they can disappear under an arbitrarily small perturbation to the tensor field. Given the instability in triple degenerate points, Zheng and Pang define the topology of a 3D symmetric tensor field in terms of {\em double degenerate points}, where the tensor field has two equal eigenvalues (repeating) and a third, distinct eigenvalue (non-repeating). Such features not only are structurally stable but also form curves. Zheng et al.~\cite{Zheng:05a} extract degenerate curves by using numeric methods on the levelsets of a degree-six discriminant. Improving on this technique, Tricoche et al.~\cite{Tricoche:08} point out that degenerate curves are part of the ridge and valley lines of a tensor invariant, {\em tensor mode}. Thus, they extract degenerate curves more robustly by reusing techniques for extracting ridge and valley lines. Palacios et al.~\cite{Palacios:17} identify a number of degenerate curve editing operations for 3D symmetric tensor fields.

Palacios et al.~\cite{Palacios:16} introduce the notion of neutral tensors and incorporate neutral surfaces into tensor field topology. Roy et al.~\cite{Roy:19} provide a parameterization for the set of degenerate tensors and a parameterization for the set of neutral tensors in a 3D piecewise linear tensor field, which they use to robustly extract degenerate curves and neutral surfaces at any given accuracy. Qu et al.~\cite{Qu:21} extend these parameterizations to seamlessly and robustly extract mode surfaces (an extension of degenerate curves and neutral surfaces).

Despite these advances, the understanding of 3D symmetric tensor field topology is still rather fragmented in that the topology is treated as a collection of isolated curves and surfaces. Relatively little attention is given to the topological structure of individual objects (\eg whether a degenerate curve is a loop and forms a knot) and the interactions among these objects (\eg whether two degenerate curves form a link). We address these with a topological graph that provides a more complete and global picture of the tensor fields, to be described next.

In their pioneering research, Zheng et al.~\cite{Zheng:05b} define the notions of wedges and trisectors for 3D symmetric tensor fields. They show that near a degenerate point, the projection of the tensor field onto the repeating plane at the point exhibit 2D degenerate point patterns such as wedges and trisectors. They further point out that between the wedge segments and trisector segments along a degenerate curve are transition points whose dominant eigenvectors are perpendicular to the tangent of the degenerate curve. Since then, there has been relatively little follow-up research on wedges/trisectors for visualization applications to engineering datasets. Zhang et al.~\cite{ZhangY:17b} study the physical meanings of wedges and trisectors in {\em 2D} symmetric tensor fields. In this paper, we address the issue by studying how the wedge/trisector classification can aid both local and global 3D tensor field analysis.

3D linear tensor fields are the simplest tensor fields as the tensor values are linear with respect to the coordinates. There has been some exploration on the topology of such fields, including the findings that there are at most four degenerate curves~\cite{ZhangY:17a} and at most eight transition points~\cite{ZhangY:20} in such fields.

There has been work on using a topological graph for 2D vector fields and tensor fields, such as the {\em Morse Connection Graphs} for vector fields~\cite{Chen:07} and the {\em eigenvalue graphs} and {\em eigenvector graphs} for 2D asymmetric tensor fields~\cite{Lin:12}. Tao et al.~\cite{tao2017semantic} apply graph analysis techniques from information visualization to flow visualization by introducing the notion of {\em semantic flow graphs}. The nodes of the graphs can be the aggregations of streamlines, regions of certain characteristics, and singularities in the field. In our work, we focus on topological features in the field such as degenerate curves and regions bounded by neutral surfaces. Hyperstreamlines are not part of the graph. In addition, each such feature is given its own node in the graph.

\section{Mathematical Background}
\label{sec:math_background}
We first review relevant concepts and results regarding 3D symmetric tensor field topology~\cite{Zheng:04,Zheng:05a,Palacios:16,Roy:19,Qu:21}. A 3D symmetric tensor has three real-valued eigenvalues:
\begin{inparaenum}[(i)]
	\item the {\em major eigenvalue} (the largest),
	\item the {\em medium eigenvalue}, and
	\item the {\em minor eigenvalue} (the smallest).
\end{inparaenum}
Eigenvectors corresponding to the major eigenvalue are referred to as the {\em major eigenvectors}. We can also define {\em medium eigenvectors} and {\em minor eigenvectors} in a similar fashion.

The sum of the eigenvalues is the {\em trace} of the tensor, while the product of the eigenvalues is the {\em determinant}. Treating the eigenvalues as a vector, its vector magnitude is used to define the {\em magnitude} of the original tensor, which is the Frobenius norm of the tensor. Given a tensor $T$, its trace is denoted by $\trace{(T)}$, its determinant by $|T|$, and its magnitude $\Vert T \Vert$.

A tensor $T$ can be uniquely decomposed into the sum of $D$, a multiple of the identity tensor, and $A$, a symmetric tensor with a zero trace.
Here, $A$ is referred to as the {\em deviator} of $T$.
Note that $T$ and $A$ have the same set of eigenvectors.
That is, a vector $v$ is an eigenvector of $T$ if and only if $v$ is an eigenvector of $A$.
Consequently, when discussing topological properties of a tensor field where eigenvector analysis is the central theme, it is usually sufficient to focus on its deviator $A$.
For example, the {\em mode} of a tensor $T$ is defined in terms of the magnitude and determinant of its deviator as
 $\mu(T) = 3\sqrt{6}\frac{|A|}{||A||^3}$.

The modes of 3D symmetric tensors have a range of $[-1, 1]$. When $\mu(T) > 0$, the deviator $A$, whose eigenvalues sum to zero, has one positive eigenvalue and two negative eigenvalues. The positive eigenvalue, \ie the major eigenvalue, is referred to as the {\em dominant eigenvalue}~\cite{Qu:21}. The major eigenvectors are referred to as the {\em dominant eigenvectors}. In solid mechanics, this case corresponds to a volume-preserving deformation with two principal axes of compression (negative eigenvalues) and one principal axis of extension (positive eigenvalue). Such tensors are referred to as {\em linear tensors}. When $\mu(T)=1$, the two negative eigenvalues are equal, indicating isotropic compression in the plane perpendicular to the principal axis of extension (eigenvector corresponding to the positive eigenvalue). Note that a tensor with at least two equal eigenvalues is referred to as being a {\em degenerate tensor}. Thus, a tensor of mode $1$ is a linear degenerate tensor, which corresponds to {\em uniaxial extension} in solid mechanics~\cite{CRISCIONE:00}. Similarly, when $\mu(T) <0$, the deviator has two positive eigenvalues and one negative eigenvalue, indicating two principal axes of extension and one principal axis of compression. Such tensors are referred to as {\em planar tensors}. The negative eigenvalue, \ie the minor eigenvalue, is the {\em dominant eigenvalue} in this case. Consequently, the minor eigenvectors are the dominant eigenvectors. When $\mu(T)=-1$, $T$ is a planar degenerate tensor and corresponds to {\em uniaxial compression} in solid mechanics~\cite{CRISCIONE:00}. When $\mu(T)=0$, its deviator $A$ has one positive eigenvalue, one zero eigenvalue, and one negative eigenvalue that has the same magnitude as the positive eigenvalue. In this case, $T$ is referred to as being a {\em neutral tensor}. In solid mechanics, a neutral tensor corresponds to {\em pure shear}~\cite{CRISCIONE:00}.

A tensor field is a tensor-valued function over its domain. The topology of a tensor field consists of its degenerate points (where the tensor value is a degenerate tensor) and neutral points (where the tensor value is a neutral tensor). Under structurally stable conditions, the set of degenerate points forms curves (degenerate curves), and the set of neutral points forms surfaces (neutral surfaces). Figure~\ref{fig:tensor_topology_illustration} (a) shows a tensor field with its degenerate curves (the colored curves) and neutral surfaces (chartreuse). Along a degenerate curve, the degenerate points are either all linear (green or blue) or all planar (yellow or red). We refer to this as the {\em linearity/planarity classification} of degenerate points. Note that the neutral surface divides the domain into linear regions, where linear degenerate curves reside, and planar regions, where planar degenerate curves reside. Between degenerate curves and neutral surfaces are {\em mode surfaces} (Figure~\ref{fig:tensor_topology_illustration} (a): the cyan and orange surfaces), which are the isosurfaces of the tensor mode. A linear region is thus a connected component of the union of all positive mode surfaces, and a planar region is a connected component of the union of all negative mode surfaces.

\begin{figure}[!t]
	\centering
	\begin{overpic}[width={0.7\columnwidth}]{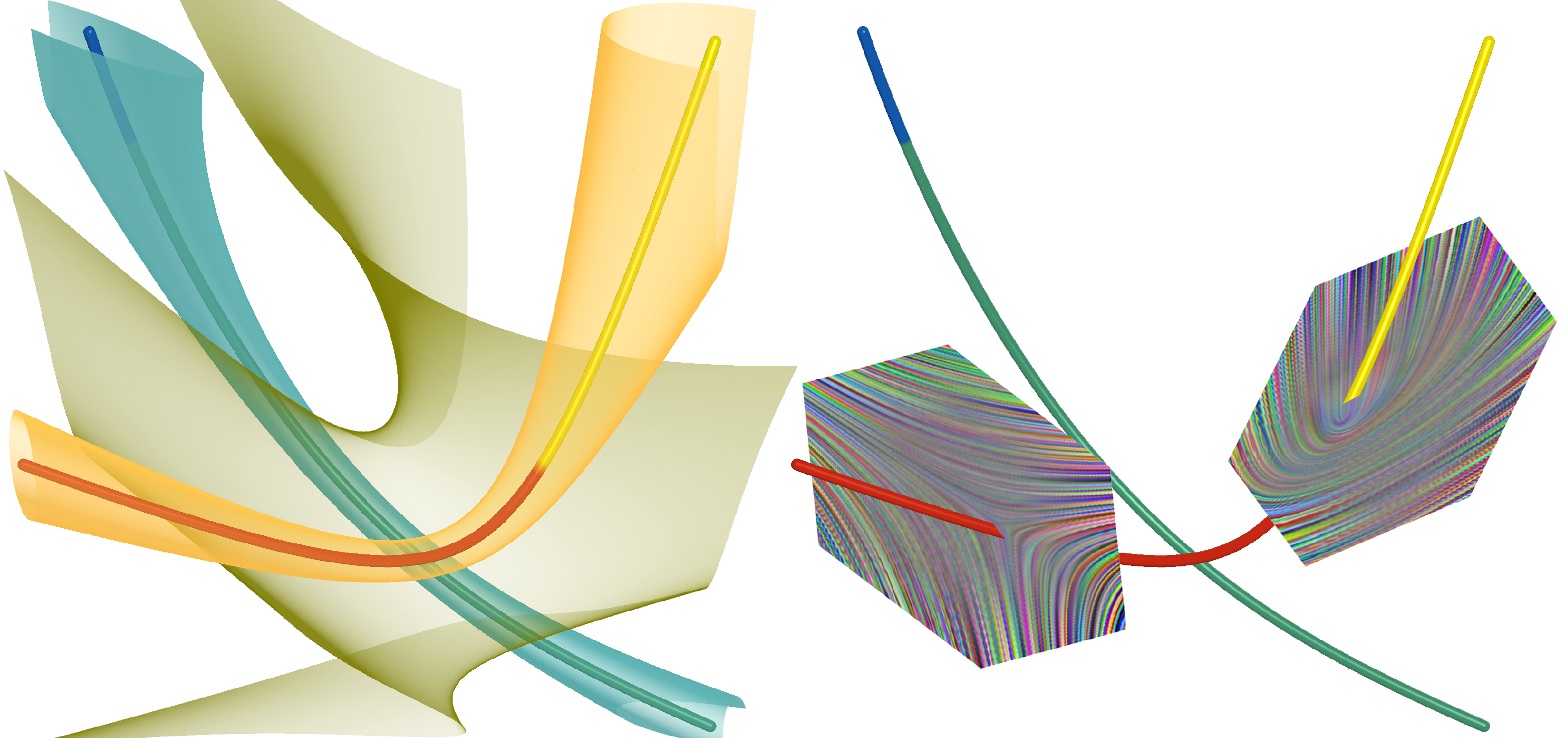}
		\put( 45, -7) {\small(a)}
		\put(130, -7) {\small(b)}
	\end{overpic}
	\caption{Degenerate curves ((a): the colored curves) and neutral surfaces ((a): the chartreuse surface) are both special mode surfaces. Additional mode surfaces (orange and cyan surfaces) occur between topological features.  Along a degenerate curve (b), the projected tensor fields onto the repeating planes show 2D degenerate patterns such as the wedge (the plane intersecting the yellow curve segment) and the trisector (the plane intersecting the red curve segment).
	}
	\label{fig:tensor_topology_illustration}
\end{figure}

In addition to the linearity/planarity classification, a degenerate point can be further classified based on the tensor index of the tensor field projected onto the repeating plane, i.e., the plane perpendicular to the dominant eigenvector at the degenerate point. A degenerate point in a 3D symmetric tensor field is referred to as a {\em wedge} if the same point is a wedge in the projected 2D tensor field (Figure~\ref{fig:tensor_topology_illustration} (b): the plane intersecting the yellow curve segment). Similarly, a trisector degenerate point in the 3D tensor field is also a trisector degenerate point in the 2D projected tensor field onto the repeating plane (Figure~\ref{fig:tensor_topology_illustration} (b): the plane intersecting the red curve segment). The wedge/trisector classification along a degenerate curve is not always the same. The points between a wedge segment and a trisector segment along a degenerate curve ((Figure~\ref{fig:tensor_topology_illustration} (b): between yellow and red segments or between green and blue segments) are referred to as the {\em transition points}. Combining the two classifications, a degenerate point can be classified as a linear wedge (green), a linear trisector (blue), a planar wedge (yellow), a planar trisector (red), a linear transition point (between green and blue segments), or a planar transition point (between yellow and red segments).

\section{Topological Graphs}
\label{sec:topological_graphs}

In this section, we provide more detail on the various components of our topological graph as well as the motivation behind our visual design.

\subsection{Regions of Uniform Linearity/Planarity}

When crossing from a linear region into the planar region via the neutral surface, the dominant eigenvector field switches from the major eigenvector field to the minor eigenvector field, with discontinuity. Consequently, we consider a decomposition of the domain into connected regions of purely linear tensor behaviors and purely planar tensor behaviors. Any pair of adjacent regions must consist of one linear region and one planar region, separated by the neutral surface.

Linear regions and planar regions have vastly different physical behaviors, and their interplay is a reflection and direct result of the boundary condition of the simulation, the shape of the domain, and the distribution of the material. A large linear region may have many small pockets of planar regions inside. In solid mechanics, this can indicate a nonuniform distribution of material deformation behavior. To quantify whether this nonuniform material behavior could affect product life is of prominent interest to design engineers. A planar region may border a linear region through multiple sheets of neutral surfaces, indicating the complex topology of the computational domain such as a mechanical part with multiple handles. Each such region can have complex geometric and topological structures as shown in Figure~\ref{fig:graph}. These structures reflect the behaviors of the underlying tensor fields, which we wish to capture. One measure for the topological complexity of a region $R$ is its {\em homology}~\cite{Kaczynski:2004}, which consists of a family of {\em groups} $\{ H_i(R) \| i \in \mathbb{Z}, i\ge 0\}$. Geometrically speaking, each generator of $H_i(R)$ represents an $i$-dimensional hole in $R$. The first Betti number, $\beta_1$, is the number of one-dimensional holes in $R$, while the second Betti number, $\beta_2$, is the number of two-dimensional holes, or voids, in $R$. One can think of the two-dimensional holes as air bubbles trapped by $R$. That is, each air bubble region has only one neighbor, which is $R$ itself. Note that each void in $R$ is itself a region of uniform linearity/planarity. Furthermore, if $R$ is a linear region, then each void trapped by $R$ must be a planar region and vice versa. Therefore, $\beta_2$ highlights the adjacency interaction between linear regions and planar regions.

The larger the Betti numbers, the more complicated the geometry is for the region as well as more interactions with other regions. Figure~\ref{fig:graph} shows an example in which the domain is a solid torus and there are three regions, two of which are planar (orange) and one linear (green). The innermost region is planar, which also has the shape of a solid torus. Such a shape has no two-dimensional holes, \ie $\beta_2=0$, and a single one-dimensional hole (the meridian of the torus), \ie $\beta_1=1$. The outermost region (Figure~\ref{fig:graph}(a)) is a thin layer of planar region $\beta_1=4$. The linear region (Figure~\ref{fig:graph}(b)) is bounded from inside by the innermost planar region and from outside by the boundary of the domain and the thin planar region. It has one bubble inside (the innermost planar region) and two one-dimensional holes (one due to the smaller torus and one due to the hole in the domain itself). Thus, $\beta_2=1$ and $\beta_1=2$.

\subsection{Indices and Network of Degenerate Curves}
\label{sec:index}

A degenerate curve can be an open curve, \ie touches the boundary of the domain, or a closed loop. Moreover, degenerate curves do not live in isolation for 3D symmetric tensor fields. This is in sharp contrast to the 2D case, in which the set of degenerate points is isolated under structurally stable conditions. Each degenerate point in 2D tensor fields can be measured in terms of its {\em tensor index}~\cite{Zhang:07} defined as follows: when traveling along a loop enclosing the degenerate point in a counterclockwise fashion, the unit major eigenvector field along the loop also covers a circle (the Gauss circle) a number of times in which the number, a multiple of $\frac{1}{2}$, is the tensor index. The fundamental degenerate points include wedges (index $\frac{1}{2}$) and trisectors (index $-\frac{1}{2}$). Note that the sign refers to whether the unit eigenvector field travels along the Gauss circle counterclockwise (wedges) or clockwise (trisectors).

The set of degenerate points can be complicated for a 3D tensor field. For example, a degenerate curve can be a loop and even form a knot. Furthermore, two degenerate loops can be linked even when they belong to different regions. To better understand the relationships among the curves in the degenerate curve network, we define a topological characterization of degenerate curves, their indexes, in the following paragraphs.

Let $R$ be a topological disk without self-intersections such that there are no degenerate points on its boundary $\partial R$ (the circles in Figure~\ref{fig:thm_2} (a-b)). We consider the right-handed frames formed by the unit major eigenvector $v_1$, the medium eigenvector $v_2$, and the minor eigenvector $v_3$ of the tensor fields on $\partial R$. Note that at each point $p$ where the eigenvectors are well-defined, i.e., not a degenerate point, there are four ways of selecting a right-handed frame from the eigenvectors. Let $f_0(p)=(v_1, v_2, v_3)$ be one such frame. Then $f_1(p)=(v_1, -v_2, -v_3)$, $f_2(p)=(-v_1, v_2, -v_3)$, and $f_3(p)=(-v_1, -v_2, v_3)$ are the other choices of such frames (Figure~\ref{fig:thm_2} (left)). Let $r_m$ ($0\le m \le 3$) be the 3D rotation that maps the $X$-axis to the major eigenvector in $f_m(p)$, the $Y$-axis to the medium eigenvector, and the $Z$-axis to the minor eigenvector (Figure~\ref{fig:thm_2} (left)). Using the matrix representation, $f_m(p)$ can be expressed as a special orthogonal matrix $r_m$. Define $r_x$, $r_y$, and $r_z$ as the $180^\circ$ rotation around the $X$-, $Y$-, and $Z$-axis, respectively. Then we have

\begin{eqnarray}
  r_1 &=& r_0 r_x \\
  r_2 &=& r_0 r_y \\
  r_3 &=& r_0 r_z
\end{eqnarray}

We choose $p_0 \in \partial R$ and travel along $\partial R$ for one round in order to inspect the behavior of the {\em continuous} eigenframe that is initially set to be $f_0(p_0)$ (Figure~\ref{fig:thm_2} (a-b)). Since the tensor field is continuous over $R$ and there is no degenerate point on $\partial R$, we know that the eigenvector fields are also continuous over $\partial R$. Therefore, when returning to $p_0$ after a full boundary walk, the frame $f'(p_0)$ must be $f_m(p_0)$ for some $0\le m\le 3$. That is,

\begin{equation}
  r'(p_0)=r_0(p_0) c
\end{equation}

\begin{figure}[!t]
	\centering
	\begin{overpic}[width={\columnwidth}]{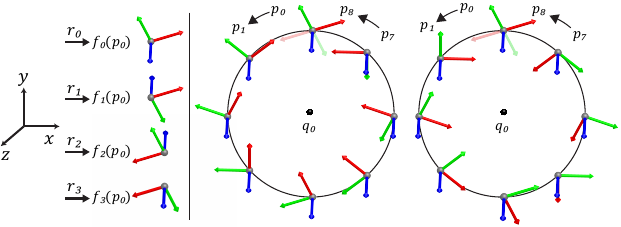}
		\put(120, -7) {\small(a)}
		\put(200, -7) {\small(b)}
	\end{overpic}
	\caption{Given a region containing a degenerate point ${\bf q}_0$ in its interior (the disk in (a) and (b)), the winding number of the boundary curve (the circle in (a-b)) is the total rotation of the frames formed by the major eigenvectors (red), medium eigenvectors (green), and minor eigenvectors (blue). It is $\mathbf{k}$ when ${\bf q}_0$ is planar wedge (a) and $\mathbf{-k}$ when ${\bf q}_0$ is planar trisector (b). Due to the continuity in the eigenvector fields away from ${\bf q}_0$, after travelling one round along the circle the eigenframe must be one of the four configurations (left).
	}
	\label{fig:thm_2}
\end{figure}

\noindent where $c=1$, $r_x$, $r_y$, or $r_z$. Note that $c=r_0(p_0)^{-1}r'(p_0)$. We can show that $c$ is a property of the loop $\partial R$ as it is independent of the choice of the initial frame at $p_0$ (Appendix: Lemma~\ref{lemma:free_initial_frame}), the choice of starting point $p_0 \in \partial R$ (Appendix: Lemma~\ref{lemma:free_starting_point}), and the direction of travel (Appendix: Lemma~\ref{lemma:free_orientation}). This indicates that the quantity $c$ is a well-defined characteristic of both the region $R$ and its boundary $\partial R$. Since $c$ is a 3D rotation which can be represented as a unit quaternion, we refer to its quaternion representation (which we still refer to as $c$) as the {\em winding number} of $R$ and $\partial R$.

We now consider a degenerate point ${\bf q}_0$ and simply-connected regions $R$ that contains it. It turns out that there is a sufficiently small neighborhood $R'$ inside which any simply-connected region containing ${\bf q_0}$ and without self-intersection has the same winding number (Appendix: Theorem~\ref{thm:winding_number_local}). This winding number is $\mathbf{i}$ if ${\bf q}_0$ is a linear wedge, $-\mathbf{i}$ if ${\bf q}_0$ is a linear trisector, $\mathbf{k}$ if ${\bf q}_0$ is planar wedge, and $\mathbf{-k}$ if ${\bf q}_0$ is planar trisector. We thus refer to  this winding number as the {\em index} of ${\bf q}_0$, which we denote by $\phi({\bf q}_0)$.

It is worth noting that in any arbitrarily small neighborhood of a transition point, it is possible to find two loops that have opposite winding numbers, e.g., $\mathbf{i}$ and $-\mathbf{i}$ for linear transition points.

Furthermore, we can show that the analysis can be made more global in the following sense. Consider a region $R$ that is free of self-intersection and contains only one degenerate point ${\bf q_0}$ in its interior. Then, if the normal to the surface $R$ is nowhere perpendicular to the dominant eigenvector field (major eigenvector in linear-dominant region and minor eigenvector in planar-dominant region), then the winding number of the boundary $\partial R$ is the same as the index of the degenerate point. (Appendix: Corollary~\ref{coro:global}).

We now consider a linear degenerate loop $\gamma$, over which the minor eigenvectors are not defined. Consequently, the notion of winding number does not apply to $\gamma$. However, $\gamma$ has a sufficiently small  neighborhood $K$ that is homotopically equivalent to $\gamma$. Consider two simple non-contracting loops $\eta_1, \eta_2  \subset K$ and a ring $\psi \subset K$ bounded by $\eta_1$ and $\eta_2$.
\begin{wrapfigure}{r}{0.2\columnwidth}
	\vspace{-2em}
	\begin{center}
		\includegraphics[width=0.2\columnwidth]{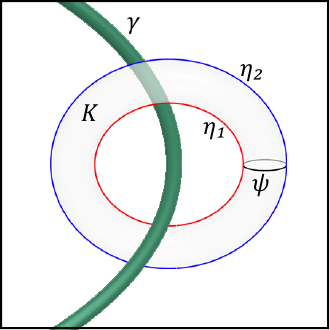}
	\end{center}
	\vspace{-2em}
\end{wrapfigure}
Assume that $\psi$ does not intersect the degenerate loop $\gamma$. Then by using techniques similar to the one in the proof for Corollary~\ref{coro:global}, we can show that the winding number of the region $\psi=1$, i.e., degenerate point-free. This implies that the winding number of $\eta_1$ is equal to that of the $\eta_2$. Since the choice of $\eta_1$ and $\eta_2$ is arbitrary, we can choose them to be arbitrarily close $\gamma$. Therefore, we can define the winding number for $\gamma$ to be that of its neighboring loops.

Note that the winding number of a degenerate loop $\gamma$ is not the index of $\gamma$. Rather, it is the index of a degenerate loop linked to $\gamma$. Suppose that $\gamma$ is the boundary of a topological disk without self-intersection. When the winding number of $\gamma$ is $\pm \mathbf{k}$, there must be a planar degenerate curve $\rho$ that intersects the topological disk. If $\rho$ is also a degenerate loop, then $\gamma$ and $\rho$ form a link. Similarly, when the winding number of $\gamma$ is $\pm \mathbf{i}$, there must be a linear degenerate curve $\rho$. When $\rho$ is also a loop, then $\gamma$ and $\rho$ also form a link. Note that $\gamma$ and $\rho$ may not be part of the same region even when they are of the same linearity/planarity type. In this case, their container regions have a relationship that is different from the adjacency relationship, even when the two regions are not adjacent.

\subsection{Visual Design of Topological Graph}
\label{sec:graph_rationale}

In this section, we provide the rationale behind the visual design of our graph, which aims to minimize edge crossing and accentuate the nodes and their properties.

In our topological graph, each region is represented by a node with the shape of a square. A linear region is colored cyan and a planar region is colored yellow. A pair of adjacent regions have their corresponding nodes connected with an edge in the graph. Furthermore, if a region is inside another region, we add a rectangular glyph on the edge to highlight the containment relationship. Since no regions of the same type can be adjacent to each other, we place all linear regions (cyan squares) on one row and all planar regions (yellow squares) on the row above. This allows the user to easily locate a region if its linearity/planarity is known. We show the Betti numbers $\beta_1$ and $\beta_2$ of a region by writing their values inside the square in the graph that corresponds to the region. To differentiate between the two values, we show $\beta_1$ next to an ellipse and $\beta_2$ next to an ellipsoid. Since $\beta_2$ records the number of air bubbles (the other type of regions) inside the region, it is always written so that it is closer to the row for the other types of regions. That is, $\beta_2$ is written in the top row of the square for linear regions (cyan) and the bottom row for planar regions (yellow). When $\beta_1=\beta_2=0$, they are not written to make it easier to identify such regions, which are contractible. Lastly, in our graph, we sort the nodes by the Betti numbers and the volume of their corresponding regions with an ascending order for the linear regions and a descending order for the planar regions, aiming to reduce the number of crossing points between edges in the graph.

Each degenerate curve is contained entirely inside a region of the same linearity/planarity. In the graph (Figure~\ref{fig:graph}), every degenerate curve has its node in the graph connected by an edge to the node that represents its container region.  To reduce the number of unnecessary crossings among this type of edges in the topological graph, all linear degenerate curves are placed on the row below the linear regions, and all planar degenerate curves are placed on the row above the planar regions. Degenerate curves belonging to the same region are displayed as a group of nodes. We sort the degenerate curves in the same region by their Writhe numbers, the total linking numbers, and lengths. We use closed rings and half rings to indicate degnerate loops and open curves, respectively. On the rings, we color the linear wedge on the curve in green, the linear trisector in blue, the planar wedge in yellow, and the planar trisector in red. For knotted loops, the Writhe number is enclosed inside the ring. Furthermore, when the Jones polynomial of a degenerate loop is not a constant, we regard the degenerate loop as a knot and add $\ast$ to the corresponding node. Linked degenerate curves have an edge connecting their nodes in the graph, with the linking number written next to the edge.

\section{Graph Construction}
\label{sec:graph_construction}

In this section, we describe our pipeline to construct the topological graph given a 3D tensor field. The input to our pipeline is a 3D piecewise linear tensor field defined on a tetrahedral mesh representing some physical domain. The tensor values are given at the vertices in the mesh and are linearly interpolated per tetrahedron. Such an interpolation ensures the continuity of the tensor field and thus its topological features across the faces of the mesh.

Our pipeline consists of the following stages (Figure~\ref{fig:graph_construction}):
\begin{inparaenum}[(i)]
	\item degenerate curve and neutral surface extraction,
    \item region extraction and processing, and
	\item degenerate curve processing.
\end{inparaenum}
We describe each of the stages in detail next.

{\bf Degenerate Curve and Neutral Surface Extraction:} We first seamlessly extract the degenerate curves in the field using the technique of Roy et al.~\cite{Roy:19}, which computes the set of degenerate points inside each tetrahedron using a parameterization that maps all degenerate points in a linear tensor field to an ellipse. Degenerate curves from adjacent tetrahedra are then stitched together across their common boundary faces. In addition, this technique classifies the linearity/planarity of each degenerate curve. Additionally, we compute the wedge/trisector classification for each sample point in the polyline representing the degenerate curve. We also mark for each degenerate curve whether it is an open curve or a loop.

In parallel, we extract the neutral surfaces in the tensor field using the technique of Qu et al.~\cite{Qu:21}, which extracts the neutral surface from each tet based on a parameterization that maps all neutral points in a linear tensor field to the Projective plane with a handle attached. The pieces from individual tets are then stitched together across the tets' boundaries.

\begin{figure}[!t]
	\centering
	\begin{overpic}[width={0.95\columnwidth}]{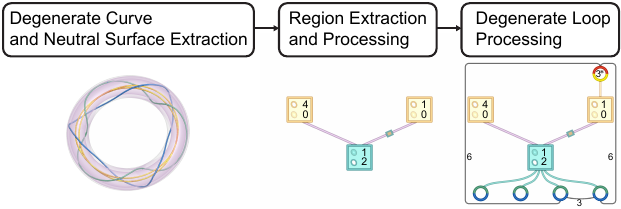}
	\end{overpic}
	\caption{
		The pipeline of our topological graph construction algorithm.
	}
	\label{fig:graph_construction}
\end{figure}

{\bf Region Extraction and Processing:} As the regions are bounded by neutral surfaces, we create linear regions and planar regions by dividing each tetrahedron into sub-regions and connecting adjacent sub-regions in neighboring tets with the same linearity or planarity. Figure~\ref{fig:region_extraction} illustrates the process for planar regions. For each tet, using the neutral surfaces inside the tet and the intersection curves on the tet faces (Figure~\ref{fig:region_extraction} (b)), we determine the boundary surfaces of the sub-regions by segmenting the tet faces with the intersection curves into patches (Figure~\ref{fig:region_extraction} (c)). We identify the linearity/planarity of each patch based on its vertex tensor value or adjacent patches. Note the adjacent patches must have opposite linearity/planarity classifications. We stitch the neutral surfaces and the patches across their shared edges to create the boundary surface of planar sub-regions (Figure~\ref{fig:region_extraction} (d)). Inside the boundary surface, the sub-region has uniform linearity/planarity. Finally, we trace the planar region by finding the connected sub-regions over the mesh (Figure~\ref{fig:region_extraction} (e)). Linear regions are extracted in a similar fashion. A node is created for each region in the graph and visualized as a colored square.

Next, we go through each triangle in the neutral surface and identify the pair of linear and planar regions on both sides. If two regions share an open sheet of the neutral surface, they are adjacent to each other. On the other hand, if two regions share a closed sheet of the neutral surface, one region is inside the other region. An edge is created for each pair of adjacent regions, while the containment relationship is further highlighted with a rectangle on the edge in the graph.

For each region, we compute its volume using the technique in~\cite{crane:2018:discrete}. To compute the Betti numbers of the region, recall that it is a $3$-manifold in the $XYZ$ space bounded by the neutral surfaces. Thus, $\beta_0$ of a given region $R$ is one and $\beta_i$ is zero for $i > 2$. To compute the second Betti number $\beta_2$ of the given region $R$, we identify all the regions that are adjacent to $R$ and share a closed border with $R$. Such a region is a void trapped by $R$ and contributes one generator for the second homological group $H_2$ of $R$. Thus, $\beta_2$ of $R$ is the number of such regions. According to~\cite{rotman:2013:introduction}, $\beta_1(R) = \beta_0(R) + \beta_2(R) - \chi(R)$ where $\chi(R)$ is the Euler characteristic of $R$. Note that $\beta_0(R) = 1$ in our cases. Since a region can be represented as a $3$-simplicial complex, the Euler characteristic of a region is defined as $\chi(R)=V-E+F-T$ where $V$ is the number of vertices, $E$ the number of edges, $F$ the number of faces, and $T$ the number of tets in $R$. However, tetrahedralizing a region can be time-consuming. We propose an effective evaluation of the $\chi(R)$ using the result from~\cite{tom:2008:algebraic}, which states that $\chi(M)$, the Euler characteristic of a compact $(n+1)$-manifold $M$, is related to $\chi(\partial M)$, the Euler characteristic of its boundary $\partial M$, by $\chi(\partial M)=(1+(-1)^n)\chi(M)$. Since our regions are compact $3$-manifolds, we have $2\chi(R)=\chi(\partial R)$ where $\partial R$ is the boundary surface of the given region. Since $\partial R$ consists of disjoint sheets of neutral surfaces $\{S_i\}_{i=1}^n$, we have

\begin{equation}
	\chi(R)= \frac{1}{2} \chi(\partial R) = \frac{1}{2}\sum_{i=1}^n \chi (S_i).
\end{equation}

Recall that the Euler characteristic of the surface $S_i$ is defined as $\chi(S_i)=V_i-E_i+F_i$. Here, $V_i$ is the number of vertices of $S_i$, $E_i$ is the number of edges, and $F_i$ is the number of faces. This allows us to compute $\chi(R)$ without tetrahedralizing $R$.

Once $\beta_1$ and $\beta_2$ are computed, we visualize them for each region using the scheme described in \sref{sec:topological_graphs}. Furthermore, we sort the nodes of the planar/linear regions in ascending/descending order with their $\beta_1$'s and volumes.

\begin{figure}[!t]
	\centering
	\begin{overpic}[width={\columnwidth}]{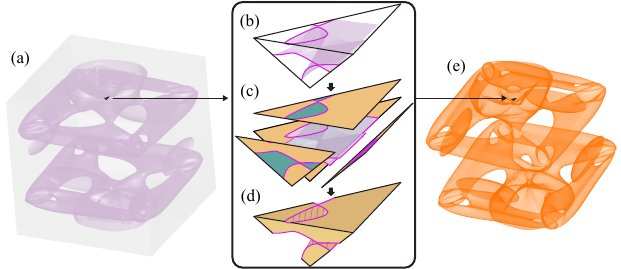}
	\end{overpic}
	\caption{This figure illustrates the extraction of planar regions. Given the neutral surfaces of the tensor field (a), we extract the neutral curves inside each face of the mesh (b) which we use to segment a face into patches of uniform linearity/planarity ((c): linear patches in green and planar patches in orange). The neutral surface and the patches in the mesh faces are then stitched together across shared edges (d), which form the boundary of planar regions inside each tet. These regions are then stitched together to form the regions (e) which are the nodes in our topological graph.
        }
	\label{fig:region_extraction}
\end{figure}

{\bf Degenerate Curve Processing:} To identify the container region of a degenerate curve, we first find an intersection point $p_0$ of the curve with a tet face of the mesh. Note that the face may consist of more than one patch, and the point can be inside exactly one patch. Thus, we test whether the point is inside one of the patches in the face by computing the winding number of the displacement vector field with respect to the point. That is, for each point $q$ on the boundary of the patch, we compute the vector $q-p_0$. When travelling along the patch boundary once, the directional component of $q-p_0$ must travel along the Gauss circle an integer number of times. The number is one if and only if the point $p_0$ is inside the patch. Once we have identified the patch containing $p_0$, the region that contains the patch is then marked as the container region for the degenerate curve.

We now identify knots and links among degenerate curves. The {\em Gauss linking integral} is a useful measure of the extent to which a pair of curves are linked to each other. A {\em link diagram} is a projection of a curve network onto a 2D plane. For two closed curves, the Gauss linking integral is an integer-valued {\em link invariant} indicating the linking number of the link diagram formed by the two curves. For open curves, it is a real number representing the average of half the number of crossings over the link diagrams of all possible projection directions (\ie the normal vectors of 2D planes)~\cite{panagiotou:2020:knot}. Given two degenerate curves $\gamma$ and $\rho$ represented as polylines, i.e., $\gamma=\{a_i\}_{i=1}^n$ and $\rho=\{b_j\}_{j=1}^m$, the Gauss linking integral~\cite{banchoff:1976:self} is represented as

\begin{equation}
	\mathcal{L}(\gamma, \rho) = \frac{1}{4\pi} \sum_{i = 1}^n\sum_{j = 1}^m Q(a_i, b_j).
\end{equation}

\noindent where $Q(a_i, b_i)$ is the area of the quadrangle formed by the end points of $a_i$ and $b_j$. For closed curves, if their Gauss linking integral is greater than or equal to $1$, then the curves are linked and we create an edge between their corresponding nodes in the graph. Additionally, we find that while the Gauss linking integral of two open curves is greater than $0.9$, they are entangled with each other. Thus, we also add an edge for the nodes of entangled open curves.

To evaluate the knottiness of a degenerate loop $\gamma=\{a_i\}_{i=1}^n$, we compute its Writhe number~\cite{panagiotou:2020:knot} as

\begin{equation}
	W(\gamma) = \frac{1}{2\pi}\sum_{i = 2}^n\sum_{j < i} Q(a_i, a_j).
\end{equation}

While the Writhe number measures the average of the number of crossings in all possible projections of the curve, a twisted loop such as connecting the ends of a helix can have a high value of the Writhe number without actually being knotted. Hence, we compute the Jones polynomial~\cite{livingston:1993:knot} of degenerate loops to identify knots. In fact, the Jones polynomial is defined for a collection of finite loops and can thus be used to also test whether multiple curves are linked besides testing whether one loop is knotted. In this paper, we apply Jones polynomials only to knot detection and classification.

Given a set of loops $\Gamma$, it is often projected onto a plane with self-overlaps (crossing points). Such a projection leads to a {\em link diagram} (Figure~\ref{fig:jones}: blue curves in the equations). The Jones polynomial is a {\em link invariant}, i.e., regardless of the projection plane and thus the link diagram. The Jones polynomial is defined in a recursive fashion, involving a base case and three recursive simplification rules (Figure~\ref{fig:jones}). This simplification allows us to systematically reduce the number of crossing points with either reconnection (Figure~\ref{fig:jones} (a-2.1 and a-2.2)) or removing an unknot which has no crossing with the rest of the curve network (Figure~\ref{fig:jones} (a-3)). Unfortunately, as can be seen the second simplification rule, the Jones polynomial can be the sum of $O(2^n)$ terms where $n$ is the number of crossing points in the link diagram. In fact, computing the Jones polynomial of a curve network $\Gamma$ is an NP-hard problem due to the recursive nature of its definition~\cite{livingston:1993:knot}.

\begin{figure}[!t]
	\centering
	\begin{overpic}[width={0.9\columnwidth}]{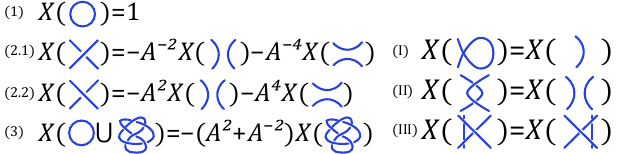}
			\put(50, -8)  {\small (a)}
			\put(180, -8)  {\small (b)}
		\end{overpic}
	\caption{The rules in the recursive definition of the Jones polynomials (a), which are invariant under the three types of Reidemeister moves (b). }
	\label{fig:jones}
\end{figure}

Given this fundamental challenge, we employ the technique of Livingston~\cite{livingston:1993:knot}, which further simplifies the computation process by simplifying the link diagram using the {\em Reidemeister moves} (Figure~\ref{fig:jones} (b-I, b-II, b-III)). When combined, these moves can convert a link diagram of a curve network to any other link diagram for the same curve network. Moreover, the Jones polynomials are maintained under these moves. Livingston~\cite{livingston:1993:knot} takes advantage of this and first simplifies the link diagram by using types I and II Reidemeister moves. Note that only these two types of moves reduce the number of crossing points in the link diagram. When no more Reidemeister moves are available to further simplify the link diagram, this technique resorts back to the simplification rules (Figure~\ref{fig:jones} (a-2.1), (a-2.2), and (a-3)). We refer our reader to Appendix~\ref{sec:jones} for more detail of this technique.

{\bf Interaction:} Given a complex dataset, the graph may contain tens of regions and degenerate curves, making it difficult for domain scientists to correlate a topological feature with its corresponding node in the graph. To address this, an interface is available for users to explore regions and degenerate curves in the field. By selecting the nodes and edges in the graph, the user can see the corresponding degenerate curve and region, two degenerate curves that are linked, or the common boundary between a pair of adjacent regions.
\section{Performance}

Our feature extraction algorithm is tested on a number of analytical and simulation data from solid mechanics. The number of tetrahedra in our data ranges from $780300$ to $1953720$. Measurements were taken on a computer with an Intel(R) Xeon(R) E3-2124G CPU$@$ 3.40 GHz, 16GB of RAM, and an NVIDIA Quadro P620 GPU. The time to extract neutral surfaces and degenerate curves averages to $10.37$ seconds and $0.47$ seconds, respectively. Region extraction and processing takes $2.3$ seconds on average. Lastly, the average time to compute the Writhe number and the Jones polynomial of degenerate loops and the linking numbers of the pairs of the degenerate curves is $1.9$ seconds on average.

\section{Applications}
\label{sec:apps}

Our topological graph allows the features in a tensor field to be visualized holistically. The user can interactively inspect a single feature such as a degenerate curve and a region as well as the relationships between two features. In addition, our topological graph allows two datasets to be compared using their topological features instead of pointwise comparisons. The datasets to be compared can be from a simulation with different boundary conditions or materials. These scenarios have applications in solid mechanics and material science.

{\bf Solid Mechanics:} Given a mechanical design, it is important to test the durability of the design under various stress conditions and detect potential fractures under extreme stress. For example, O-rings are used as a sealing solution by multiple industries~\cite{hannifin:2007:parker}. They are squeezed among distinct components of machines to prevent the occurrence of fluid or gas leaks. Therefore, it is essential to numerically simulate its behaviors under different compression. There are isotropic compression (magnitude is constant) and anisotropic compression (magnitude varies) under static and dynamic sealing. We consider the cases when the anisotropy in the magnitude of the compression force is periodic both along its circumference and in the cross-section:

\begin{equation}
	u(\theta, \phi) = (1 - \alpha) + \alpha \cos(p\theta + q\phi),
	\label{eq:traction_force}
\end{equation}

\noindent where $\theta$ and $\phi$ are respectively the longitude and meridian of the surface of the O-ring (a torus), $u(\theta, \phi)$ is the magnitude of the compression force at the point $(\theta, \phi)$ on the torus, $p$ and $q$ are respectively the periodicity of $u$ in terms of its longitude and meridian, and $\alpha$ amplifies the magnitude of anisotropy. The unit of $u$ is Newton (N). Figure~\ref{fig:oring_deformation} shows the influence of the parameters on the compression force.

The O-ring has an internal diameter of $6.07 mm$ and a width of the cross section of $1.78 mm$. When $\alpha=0$, i.e., the purely isotropic case, there is one linear degenerate loop in the stress tensor field (Figure~\ref{fig:oring_3}(a)). We consider the case when $p=3$ and $q=2$ with increasing $\alpha$'s. Figure~\ref{fig:oring_3}(b) shows the case when $\alpha=0.25$. The three-fold symmetry in the deformation of the degenerate loops and the linking number of the two loops ($3$) are due to the periodicity in the compression force when $p=3$. Despite the compression force everywhere on the boundary of the O-ring, there is now a thin ring of the planar region in the middle of the torus with the Betti numbers $\beta_1=1$ and $\beta_2=0$. Accordingly, the linear region now has a void (\ie $\beta_2=1$) due to the planar region trapped inside. As $\alpha$ increases, so does the anisotropy in the magnitude of the compression force on the boundary and the size of the planar region around the core of the torus. Figure~\ref{fig:oring_3}(c) shows the case when $\alpha=0.35$, which is also the case shown in Figure~\ref{fig:teaser}. Notice that in this case, the planar region has grown large enough to even host one knotted degenerate curve inside that loops twice, indicating the core of the O-ring is going through more compression. Notice that when $\alpha=0$ (Figure~\ref{fig:oring_3} (a)), the core of the O-ring is going through extension. The change in the stress tensor at the core is due to the stronger anisotropy in the loading condition at the boundary of the O-ring. Also, when $\alpha=0.35$, we observe two additional short degenerate loops (Figure~\ref{fig:teaser}: the two curves labelled with ``zoom in''). Given that these loops are short and do not respect the three-fold symmetry in the boundary load, we hypothesize that they are topological noise due to numerical issues. Further investigation is needed to address this.

\begin{figure}[!t]
	\centering
	\begin{overpic}[width={\columnwidth}]{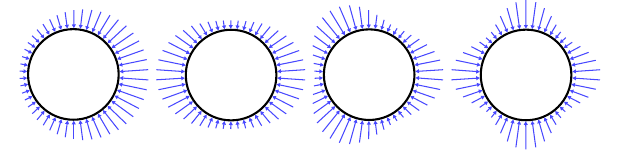}
		\put(25, -8)  {\small $q=1$}
		\put(85, -8)  {\small $q=2$}
		\put(145, -8)  {\small $q=3$}
		\put(205, -8)  {\small $q=4$}
	\end{overpic}
	\caption{We demonstrate the compression force given in Equation~\ref{eq:traction_force} at the cross-section of $\theta = 0$.
	}
	\label{fig:oring_deformation}
\end{figure}

\begin{figure}[!t]
	\centering
	\begin{overpic}[width={\columnwidth}]{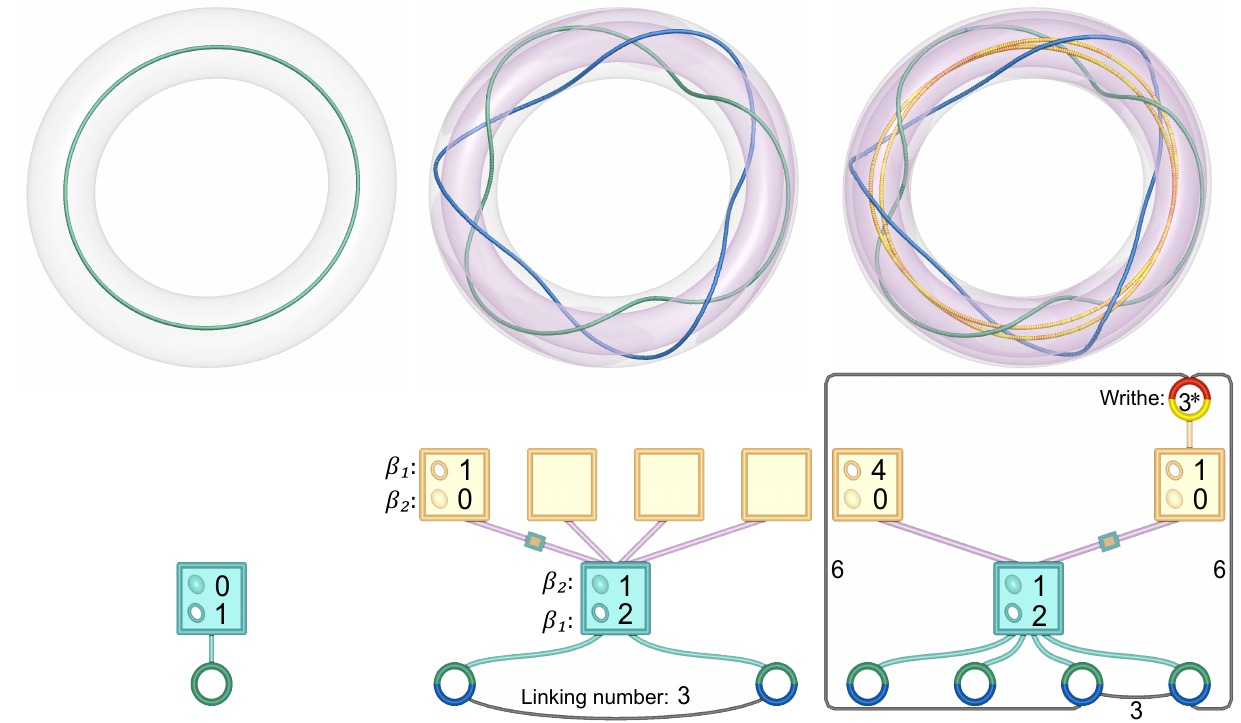}
		\put(  26, -7)  {\small(a) $\alpha = 0$}
		\put(110, -7)  {\small(b) $\alpha = 0.25$}
		\put(190, -7)  {\small(c) $\alpha = 0.35$}
	\end{overpic}
	\caption{Three O-ring scenarios with varying $\alpha$ values under the condition $p=3$ and $q=2$.
	}
	\label{fig:oring_3}
\end{figure}

\begin{figure*}[!t]
	\centering
	\begin{overpic}[width={\textwidth}]{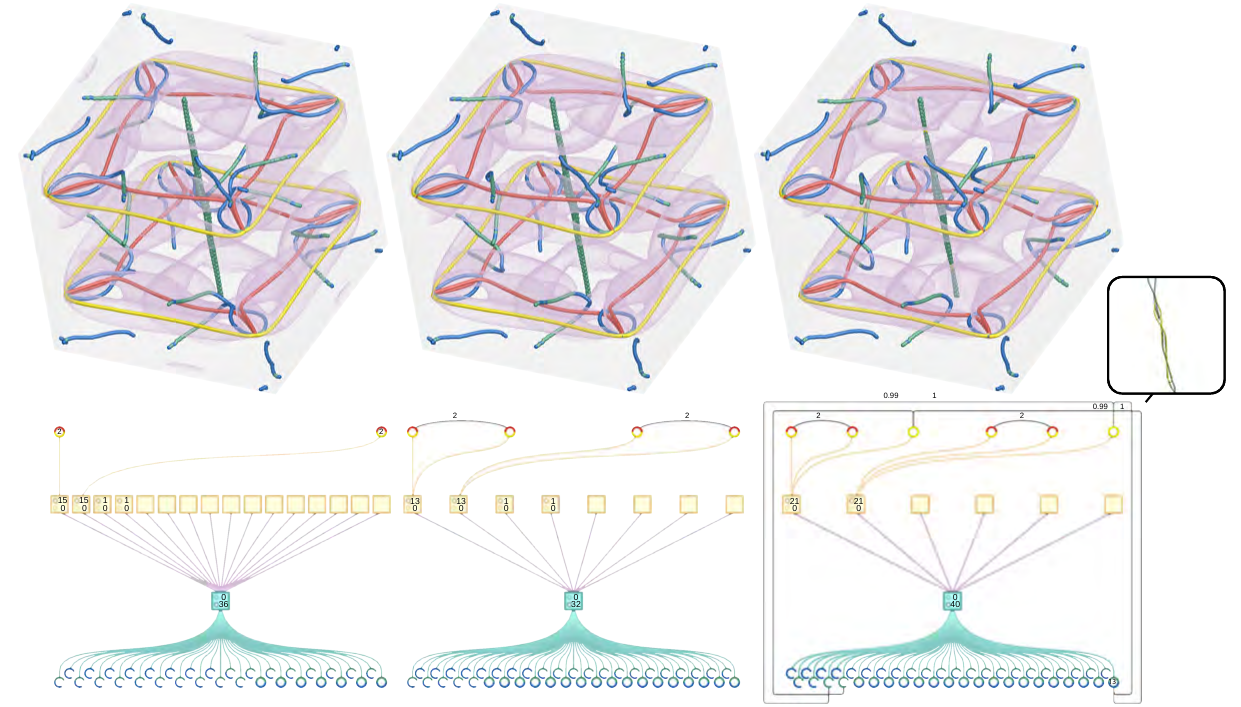}
		\put( 40, -7)  {(a) Poisson's ratio $0.13$}
		\put(200, -7)  {(b) Poisson's ratio $0.18$}
		\put(360, -7)  {(c) Poisson's ratio $0.24$}
	\end{overpic}
	\caption{Three concrete materials with different Poisson's ratios are compared in terms of the topological graphs of their stress tensor fields. We notice that in cases (a) and (b), there are four non-trivial planar regions (yellow nodes) with positive $\beta_1$ values. When the Poisson's ratio of the concrete increases to $0.24$ (c), the four planar regions become two regions near the top and the bottom of the domain and linking between degenerate curves appears (inside the zoom-in box).
	}
	\label{fig:pile1_1}
\end{figure*}

In Appendix~\ref{sec:oring_more}, we provide additional analysis of varying $p$ and $q$ values systematically while keeping $\alpha$ constant.

{\bf Material Science:} The material properties of an object have a direct impact on its response to external stress. Concrete is a material that is widely used in the construction of buildings and bridges, and its durability has a direct impact on our daily life. The Poisson's ratio is a measure of the deformation of a material in the direction perpendicular to loading~\cite{dudzinski1952young} and thus can also serve as a measure for the incompressibility of the material. For a given concrete material, due to its composition and use, the Poisson's ratio can take on a range. In this application, we consider concrete materials with the Poisson's ratios of $0.13$, $0.18$, and $0.24$, respectively. Here, we use one cubic block to represent a typical volume of a pile cap in the foundation of a building.

\begin{wrapfigure}{r}{0.3\columnwidth}
	\vspace{-2em}
	\begin{center}
		\includegraphics[width=0.3\columnwidth]{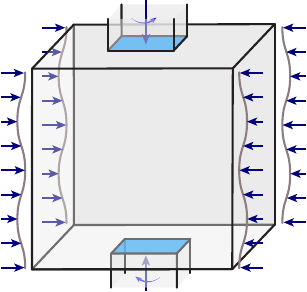}
	\end{center}
	\vspace{-2em}
\end{wrapfigure}

There is compression loading on the top, the bottom, and the sides of the block. From the sides, we impose a sinusoidal profile for the magnitude of the compression loading to represent possible interaction with neighbor blocks as illustrated in the right image. The tensor fields and their topological graphs are shown in Figure~\ref{fig:pile1_1}: (a) a Poisson's ratio of $0.13$, (b) a Poisson's ratio of $0.18$, and (c) a Poisson's ratio of $0.24$. While we expect different behaviors in the stress tensor fields within the block, we make the following observations based on the topological graphs.

There is a linear region at the center of the block that persists for all scenarios. We observe that in this linear region, there are degenerate curves connecting the top and bottom faces near the center of the domain. This is where the forces from the four sides of the block reach a balance.

On the sides of the cube, the magnitude of the load varies, which has three maximums and two minimums along a vertical line. The top and middle maximums, when coupled with the minimum inbetween them, lead to a trisector type of degenerate curve (red) near the top face of the cube. Notice the four-way symmetry in this degenerate curve. On the other hand, despite the maximal loading on the side of the cube near the top, the fact that the top surface is also being pushed down leads to a dead-end for material, namely, the wedge curve (yellow) near the top surface. Such an observation would not be possible when the wedge and trisector curves are not differentiated.

Furthermore, there are planar regions that are near the top and bottom faces of the box. The regions have a hole near the loading point at the top and bottom faces. In our graph, we show the Betti numbers of the regions, which indicate the regions' complexities. When the Poisson's ratio is $0.13$, there are many smaller planar regions scattered in the domain. When the Poisson's ratio increases to $0.18$, the material becomes less compressible. This is confirmed by the merging of the small planar regions in the topological graph. When the Poisson's ratio reaches $0.24$, the material becomes even more incompressible, which is highlighted by the merging of the planar regions, which now enclose the loading points at the top and bottom faces. In contrast, for all three materials, there is a single linear region, highlighting the fact that the side forces provide a stronger load than the loads at the top and bottom faces of the block.

We provide more examples in the supplementary video. All application examples in the paper and the video were generated using the commercial software SIMULIA~\cite{ABAQUS}.

\section{Conclusion and Future Work}
\label{sec:conclusion}

In this paper, we introduce the notion of topological graphs as a map to the global topology of 3D symmetric tensor fields. At the core of our approach are several observations and theoretical analyses of tensor field topology. In particular, we define the index for each degenerate loop and observe the existence of loops, knots, and links in the degenerate curves. We also provide algorithms to extract regions of uniform linearity/planarity and advocate the use of homology as a measure of regions' complexities.  The Writhe numbers, linking numbers, and the Jones polynomial for knots and links are also computed using existing algorithms, which can provide insight into the behavior of the underlying tensor field. Our topological graph provides a more holistic view of the topological structures in a tensor field and can be used to compare tensor fields from simulations of either different boundary conditions or different material properties. Our approach also enables the user to select individual objects (degenerate curves and regions of uniform linearity/planarity) for inspection even when these objects are occluded by other objects in the field.

In addition, we differentiate between wedges and trisectors and provide some interpretation of their physical meanings in the context of the stress tensor in engineering applications. We demonstrate how to use the topological graph for a global description of the tensor fields from solid mechanics and material science.

Our graph construction can be improved in a number of regards. For example, while we have found the use of bounding boxes for the containment relationship between two regions is sufficient for our datasets, the test can theoretically fail even when one of the regions is contained inside the other. We plan to investigate more accurate numeric methods to address this.

We consider our research one of the first steps towards a complete global topological analysis of 3D symmetric tensor fields. As such, there are many potential fruitful future research directions. For example, we wish to identify all potential bifurcations involving topological features in a 3D tensor field, thus allowing the processing of time-dependent tensor fields. Second, not all features are of equal importance, and we plan to study a multi-scale representation of tensor field topology similar to the one for 2D asymmetric tensor fields~\cite{Khan:20}. Highlighting where the graph is changing when a parameter (such as the Poisson's ratio) is varied can help improve the workflow for domain scientists, a direction we wish to pursue. In addition, we plan to explore better layouts for our topological graphs. The fact that regions of the same type cannot be connected has the potential of enabling better graph layout and thus enhanced graph readability. Including more statistics of graphs such as their Laplacian can provide additional insight into the global topology of tensor fields, and we plan to investigate this. Finally, we plan to extend our global topology approach to 3D asymmetric tensor fields, which have more types of topological features~\cite{Hung:22} and thus more complicated interactions among them.


%
%
%
%
%
\acknowledgments{The authors appreciate the constructive feedback from our anonymous reviewers. We wish to thank Peter Oliver for his help during video production. This work was supported in part by the NSF award (\# 1619383).}

\bibliographystyle{abbrv-doi-hyperref}

\bibliography{3DSymmetric_Graph}

\begin{thebibliography}{10}

\bibitem{banchoff:1976:self}
T.~Banchoff.
\newblock Self-linking numbers of space polygons.
\newblock {\em Indiana University Mathematics Journal}, 25(12):1171--1188,
  1976.

\bibitem{Calcaterra:2008}
\href{https://doi.org/10.1016/j.jmaa.2007.06.001}{C.~Calcaterra and A.~Boldt}.
\newblock \href{https://doi.org/10.1016/j.jmaa.2007.06.001}{Lipschitz flow-box
  theorem}.
\newblock \href{https://doi.org/10.1016/j.jmaa.2007.06.001}{{\em Journal of
  Mathematical Analysis and Applications}},
  \href{https://doi.org/10.1016/j.jmaa.2007.06.001}{338(2):1108--1115},
  \href{https://doi.org/10.1016/j.jmaa.2007.06.001}{2008}.
  \href{https://doi.org/10.1016/j.jmaa.2007.06.001}
{doi: {{%
10\hspace{.1pt}\discretionary{.}{%
}{.}\hspace{.4pt}1016\discretionary{/}{%
}{/}j\hspace{.1pt}\discretionary{.}{%
}{.}\hspace{.4pt}jmaa\hspace{.1pt}\discretionary{.}{%
}{.}\hspace{.4pt}2007\hspace{.1pt}\discretionary{.}{%
}{.}\hspace{.4pt}06\hspace{.1pt}\discretionary{.}{%
}{.}\hspace{.4pt}001}}}


\bibitem{EPFL-BOOK-138668}
\href{http://www.springer.com/computer/computer+imaging/book/978-1-84882-298-6}{L.~Cammoun,
  C.~A. Castano-Moraga, E.~Munoz-Moreno, D.~Sosa-Cabrera, B.~Acar,
  M.~Rodriguez-Florido, A.~Brun, H.~Knutsson, J.~Thiran, S.~Aja-Fernandez,
  R.~de~Luis~Garcia, D.~Tao, and X.~Li}.
\newblock
  \href{http://www.springer.com/computer/computer+imaging/book/978-1-84882-298-6}{{\em
  Tensors in {I}mage {P}rocessing and {C}omputer {V}ision}}.
\newblock
  \href{http://www.springer.com/computer/computer+imaging/book/978-1-84882-298-6}{Advances
  in Pattern Recognition}.
  \href{http://www.springer.com/computer/computer+imaging/book/978-1-84882-298-6}{Springer
  London},
  \href{http://www.springer.com/computer/computer+imaging/book/978-1-84882-298-6}{London},
  \href{http://www.springer.com/computer/computer+imaging/book/978-1-84882-298-6}{2009}.

\bibitem{chaimahawan:2021:finite}
\href{https://doi.org/10.1590/1679-78256290}{P.~Chaimahawan, S.~Suparp,
  P.~Joyklad, and Q.~Hussain}.
\newblock \href{https://doi.org/10.1590/1679-78256290}{Finite element analysis
  of reinforced concrete pile cap using atena}.
\newblock \href{https://doi.org/10.1590/1679-78256290}{{\em Latin American
  Journal of Solids and Structures}},
  \href{https://doi.org/10.1590/1679-78256290}{18},
  \href{https://doi.org/10.1590/1679-78256290}{2021}.
  \href{https://doi.org/10.1590/1679-78256290}
{doi: {{%
10\hspace{.1pt}\discretionary{.}{%
}{.}\hspace{.4pt}1590\discretionary{/}{%
}{/}1679\discretionary{%
}{-}{-}78256290}}}


\bibitem{Chen:07}
\href{https://doi.org/10.1109/TVCG.2007.1021}{G.~Chen, K.~Mischaikow, R.~S.
  Laramee, P.~Pilarczyk, and E.~Zhang}.
\newblock \href{https://doi.org/10.1109/TVCG.2007.1021}{Vector field editing
  and periodic orbit extraction using morse decomposition}.
\newblock \href{https://doi.org/10.1109/TVCG.2007.1021}{{\em IEEE Transactions
  on Visualization and Computer Graphics}},
  \href{https://doi.org/10.1109/TVCG.2007.1021}{13(4):769--785},
  \href{https://doi.org/10.1109/TVCG.2007.1021}{2007}.
  \href{https://doi.org/10.1109/TVCG.2007.1021}
{doi: {{%
10\hspace{.1pt}\discretionary{.}{%
}{.}\hspace{.4pt}1109\discretionary{/}{%
}{/}TVCG\hspace{.1pt}\discretionary{.}{%
}{.}\hspace{.4pt}2007\hspace{.1pt}\discretionary{.}{%
}{.}\hspace{.4pt}1021}}}


\bibitem{crane:2018:discrete}
K.~Crane.
\newblock Discrete differential geometry: An applied introduction.
\newblock {\em Notices of the AMS, Communication}, pp. 1153--1159, 2018.

\bibitem{CRISCIONE:00}
\href{https://doi.org/10.1016/S0022-5096(00)00023-5}{J.~C. Criscione, J.~D.
  Humphrey, A.~S. Douglas, and W.~C. Hunter}.
\newblock \href{https://doi.org/10.1016/S0022-5096(00)00023-5}{An invariant
  basis for natural strain which yields orthogonal stress response terms in
  isotropic hyperelasticity}.
\newblock \href{https://doi.org/10.1016/S0022-5096(00)00023-5}{{\em Journal of
  the Mechanics and Physics of Solids}},
  \href{https://doi.org/10.1016/S0022-5096(00)00023-5}{48(12):2445 -- 2465},
  \href{https://doi.org/10.1016/S0022-5096(00)00023-5}{2000}.
  \href{https://doi.org/10.1016/S0022-5096(00)00023-5}
{doi: {{%
10\hspace{.1pt}\discretionary{.}{%
}{.}\hspace{.4pt}1016\discretionary{/}{%
}{/}S0022\discretionary{%
}{-}{-}5096\discretionary{%
}{(}{(}00\discretionary{)}{%
}{)}00023\discretionary{%
}{-}{-}5}}}


\bibitem{de:1999:collapsing}
\href{https://doi.org/10.1109/VISUAL.1999.809907}{W.~De~Leeuw and
  R.~Van~Liere}.
\newblock \href{https://doi.org/10.1109/VISUAL.1999.809907}{Collapsing flow
  topology using area metrics}.
\newblock \href{https://doi.org/10.1109/VISUAL.1999.809907}{In {\em Proceedings
  Visualization '99 (Cat. No.99CB37067)}},
  \href{https://doi.org/10.1109/VISUAL.1999.809907}{pp. 349--542},
  \href{https://doi.org/10.1109/VISUAL.1999.809907}{1999}.
  \href{https://doi.org/10.1109/VISUAL.1999.809907}
{doi: {{%
10\hspace{.1pt}\discretionary{.}{%
}{.}\hspace{.4pt}1109\discretionary{/}{%
}{/}VISUAL\hspace{.1pt}\discretionary{.}{%
}{.}\hspace{.4pt}1999\hspace{.1pt}\discretionary{.}{%
}{.}\hspace{.4pt}809907}}}


\bibitem{delmarcelle:visualizing}
T.~Delmarcelle and L.~Hesselink.
\newblock {Visualizing second-order tensor fields with hyperstream lines}.
\newblock {\em IEEE Computer Graphics and Applications}, 13(4):25--33, July
  1993.

\bibitem{dudzinski1952young}
\href{https://books.google.com/books?id=GZwdyAEACAAJ}{N.~Dudzinski and R.~A.
  Establishment}.
\newblock \href{https://books.google.com/books?id=GZwdyAEACAAJ}{{\em The
  Young's Modulus, Poisson's Ratio and Rigidity Modulus of Some Aluminium
  Alloys}}.
\newblock \href{https://books.google.com/books?id=GZwdyAEACAAJ}{Number pt. 2 in
  RAE/Met-69}. \href{https://books.google.com/books?id=GZwdyAEACAAJ}{RAE},
  \href{https://books.google.com/books?id=GZwdyAEACAAJ}{1952}.

\bibitem{hannifin:2007:parker}
P.~Hannifin.
\newblock Parker o-ring handbook.
\newblock {\em Parker Hannifin Corporation, Cleveland, OH}, 2007.

\bibitem{hesselink:topology}
\href{https://doi.org/10.1109/2945.582332}{L.~Hesselink, Y.~Levy, and
  Y.~Lavin}.
\newblock \href{https://doi.org/10.1109/2945.582332}{{The topology of
  symmetric, second-order 3{D} tensor fields}}.
\newblock \href{https://doi.org/10.1109/2945.582332}{{\em IEEE Transactions on
  Visualization and Computer Graphics}},
  \href{https://doi.org/10.1109/2945.582332}{3(1):1--11},
  \href{https://doi.org/10.1109/2945.582332}{Mar. 1997}.
  \href{https://doi.org/10.1109/2945.582332}
{doi: {{%
10\hspace{.1pt}\discretionary{.}{%
}{.}\hspace{.4pt}1109\discretionary{/}{%
}{/}2945\hspace{.1pt}\discretionary{.}{%
}{.}\hspace{.4pt}582332}}}


\bibitem{Hung:22}
\href{https://doi.org/10.1109/TVCG.2021.3114808}{S.-H. Hung, Y.~Zhang, H.~Yeh,
  and E.~Zhang}.
\newblock \href{https://doi.org/10.1109/TVCG.2021.3114808}{Feature curves and
  surfaces of 3d asymmetric tensor fields}.
\newblock \href{https://doi.org/10.1109/TVCG.2021.3114808}{{\em IEEE
  Transactions on Visualization and Computer Graphics}},
  \href{https://doi.org/10.1109/TVCG.2021.3114808}{28(1):33--42},
  \href{https://doi.org/10.1109/TVCG.2021.3114808}{2022}.
  \href{https://doi.org/10.1109/TVCG.2021.3114808}
{doi: {{%
10\hspace{.1pt}\discretionary{.}{%
}{.}\hspace{.4pt}1109\discretionary{/}{%
}{/}TVCG\hspace{.1pt}\discretionary{.}{%
}{.}\hspace{.4pt}2021\hspace{.1pt}\discretionary{.}{%
}{.}\hspace{.4pt}3114808}}}


\bibitem{jankowai:2019:robust}
\href{https://doi.org/10.1111/cgf.13693}{J.~Jankowai, B.~Wang, and I.~Hotz}.
\newblock \href{https://doi.org/10.1111/cgf.13693}{Robust extraction and
  simplification of 2d symmetric tensor field topology}.
\newblock \href{https://doi.org/10.1111/cgf.13693}{In {\em Computer Graphics
  Forum}}, \href{https://doi.org/10.1111/cgf.13693}{vol.~38},
  \href{https://doi.org/10.1111/cgf.13693}{pp. 337--349}.
  \href{https://doi.org/10.1111/cgf.13693}{Wiley Online Library},
  \href{https://doi.org/10.1111/cgf.13693}{2019}.
  \href{https://doi.org/10.1111/cgf.13693}
{doi: {{%
10\hspace{.1pt}\discretionary{.}{%
}{.}\hspace{.4pt}1111\discretionary{/}{%
}{/}cgf\hspace{.1pt}\discretionary{.}{%
}{.}\hspace{.4pt}13693}}}


\bibitem{Kaczynski:2004}
T.~Kaczynski, K.~M. Mischaikow, M.~Mrozek, and K.~Mischaikow.
\newblock {\em Computational homology / Tomasz Kaczynski, Konstantin
  Mischaikow, Marian Mrozek.}
\newblock Applied mathematical sciences (Springer-Verlag New York Inc.); v.
  157. Springer, New York, 2004.

\bibitem{Khan:20}
\href{https://doi.org/10.1109/TVCG.2019.2934314}{F.~{Khan}, L.~{Roy},
  E.~{Zhang}, B.~{Qu}, S.~H. {Hung}, H.~{Yeh}, R.~S. {Laramee}, and
  Y.~{Zhang}}.
\newblock \href{https://doi.org/10.1109/TVCG.2019.2934314}{Multi-scale
  topological analysis of asymmetric tensor fields on surfaces}.
\newblock \href{https://doi.org/10.1109/TVCG.2019.2934314}{{\em IEEE
  Transactions on Visualization and Computer Graphics}},
  \href{https://doi.org/10.1109/TVCG.2019.2934314}{26(1):270--279},
  \href{https://doi.org/10.1109/TVCG.2019.2934314}{2020}.
  \href{https://doi.org/10.1109/TVCG.2019.2934314}
{doi: {{%
10\hspace{.1pt}\discretionary{.}{%
}{.}\hspace{.4pt}1109\discretionary{/}{%
}{/}TVCG\hspace{.1pt}\discretionary{.}{%
}{.}\hspace{.4pt}2019\hspace{.1pt}\discretionary{.}{%
}{.}\hspace{.4pt}2934314}}}


\bibitem{Kratz:13}
\href{https://doi.org/10.1111/j.1467-8659.2012.03231.x}{A.~Kratz, C.~Auer,
  M.~Stommel, and I.~Hotz}.
\newblock \href{https://doi.org/10.1111/j.1467-8659.2012.03231.x}{Visualization
  and analysis of second-order tensors: Moving beyond the symmetric
  positive-definite case}.
\newblock \href{https://doi.org/10.1111/j.1467-8659.2012.03231.x}{{\em Computer
  Graphics Forum}},
  \href{https://doi.org/10.1111/j.1467-8659.2012.03231.x}{32(1):49--74},
  \href{https://doi.org/10.1111/j.1467-8659.2012.03231.x}{2013}.
  \href{https://doi.org/10.1111/j.1467-8659.2012.03231.x}
{doi: {{%
10\hspace{.1pt}\discretionary{.}{%
}{.}\hspace{.4pt}1111\discretionary{/}{%
}{/}j\hspace{.1pt}\discretionary{.}{%
}{.}\hspace{.4pt}1467\discretionary{%
}{-}{-}8659\hspace{.1pt}\discretionary{.}{%
}{.}\hspace{.4pt}2012\hspace{.1pt}\discretionary{.}{%
}{.}\hspace{.4pt}03231\hspace{.1pt}\discretionary{.}{%
}{.}\hspace{.4pt}x}}}


\bibitem{Lin:12}
\href{https://doi.org/10.1007/978-3-642-23175-9_13}{Z.~Lin, H.~Yeh, R.~S.
  Laramee, and E.~Zhang}.
\newblock \href{https://doi.org/10.1007/978-3-642-23175-9_13}{{\em 2D
  Asymmetric Tensor Field Topology}},
  \href{https://doi.org/10.1007/978-3-642-23175-9_13}{pp. 191--204}.
\newblock \href{https://doi.org/10.1007/978-3-642-23175-9_13}{Springer Berlin
  Heidelberg}, \href{https://doi.org/10.1007/978-3-642-23175-9_13}{Berlin,
  Heidelberg}, \href{https://doi.org/10.1007/978-3-642-23175-9_13}{2012}.
  \href{https://doi.org/10.1007/978-3-642-23175-9_13}
{doi: {{%
10\hspace{.1pt}\discretionary{.}{%
}{.}\hspace{.4pt}1007\discretionary{/}{%
}{/}978\discretionary{%
}{-}{-}3\discretionary{%
}{-}{-}642\discretionary{%
}{-}{-}23175\discretionary{%
}{-}{-}9\_13}}}


\bibitem{livingston:1993:knot}
C.~Livingston.
\newblock {\em Knot theory}, vol.~24.
\newblock Cambridge University Press, 1993.

\bibitem{Markus:55}
\href{https://doi.org/10.2307/1970071}{L.~Markus}.
\newblock \href{https://doi.org/10.2307/1970071}{Line element fields and
  lorentz structures on differentiable manifolds}.
\newblock \href{https://doi.org/10.2307/1970071}{{\em Annals of Mathematics}},
  \href{https://doi.org/10.2307/1970071}{62(3):pp. 411--417},
  \href{https://doi.org/10.2307/1970071}{1955}.
  \href{https://doi.org/10.2307/1970071}
{doi: {{%
10\hspace{.1pt}\discretionary{.}{%
}{.}\hspace{.4pt}2307\discretionary{/}{%
}{/}1970071}}}


\bibitem{Palacios:17}
\href{https://doi.org/10.1145/3130800.3130844}{J.~Palacios, L.~Roy, P.~Kumar,
  C.~Hsu, W.~Chen, C.~Ma, L.~Wei, and E.~Zhang}.
\newblock \href{https://doi.org/10.1145/3130800.3130844}{Tensor field design in
  volumes}.
\newblock \href{https://doi.org/10.1145/3130800.3130844}{{\em {ACM} Trans.
  Graph.}},
  \href{https://doi.org/10.1145/3130800.3130844}{36(6):188:1--188:15},
  \href{https://doi.org/10.1145/3130800.3130844}{2017}.
  \href{https://doi.org/10.1145/3130800.3130844}
{doi: {{%
10\hspace{.1pt}\discretionary{.}{%
}{.}\hspace{.4pt}1145\discretionary{/}{%
}{/}3130800\hspace{.1pt}\discretionary{.}{%
}{.}\hspace{.4pt}3130844}}}


\bibitem{Palacios:16}
\href{https://doi.org/10.1109/TVCG.2015.2484343}{J.~Palacios, H.~Yeh, W.~Wang,
  Y.~Zhang, R.~S. Laramee, R.~Sharma, T.~Schultz, and E.~Zhang}.
\newblock \href{https://doi.org/10.1109/TVCG.2015.2484343}{Feature surfaces in
  symmetric tensor fields based on eigenvalue manifold}.
\newblock \href{https://doi.org/10.1109/TVCG.2015.2484343}{{\em IEEE
  Transactions on Visualization and Computer Graphics}},
  \href{https://doi.org/10.1109/TVCG.2015.2484343}{22(3):1248--1260},
  \href{https://doi.org/10.1109/TVCG.2015.2484343}{Mar. 2016}.
  \href{https://doi.org/10.1109/TVCG.2015.2484343}
{doi: {{%
10\hspace{.1pt}\discretionary{.}{%
}{.}\hspace{.4pt}1109\discretionary{/}{%
}{/}TVCG\hspace{.1pt}\discretionary{.}{%
}{.}\hspace{.4pt}2015\hspace{.1pt}\discretionary{.}{%
}{.}\hspace{.4pt}2484343}}}


\bibitem{panagiotou:2020:knot}
\href{https://doi.org/10.1098/rspa.2020.0124}{E.~Panagiotou and L.~H.
  Kauffman}.
\newblock \href{https://doi.org/10.1098/rspa.2020.0124}{Knot polynomials of
  open and closed curves}.
\newblock \href{https://doi.org/10.1098/rspa.2020.0124}{{\em Proceedings of the
  Royal Society A}},
  \href{https://doi.org/10.1098/rspa.2020.0124}{476(2240):20200124},
  \href{https://doi.org/10.1098/rspa.2020.0124}{2020}.
  \href{https://doi.org/10.1098/rspa.2020.0124}
{doi: {{%
10\hspace{.1pt}\discretionary{.}{%
}{.}\hspace{.4pt}1098\discretionary{/}{%
}{/}rspa\hspace{.1pt}\discretionary{.}{%
}{.}\hspace{.4pt}2020\hspace{.1pt}\discretionary{.}{%
}{.}\hspace{.4pt}0124}}}


\bibitem{Qu:21}
\href{https://doi.org/10.1109/TVCG.2020.3030431}{B.~{Qu}, L.~{Roy}, Y.~{Zhang},
  and E.~{Zhang}}.
\newblock \href{https://doi.org/10.1109/TVCG.2020.3030431}{Mode surfaces of
  symmetric tensor fields: Topological analysis and seamless extraction}.
\newblock \href{https://doi.org/10.1109/TVCG.2020.3030431}{{\em IEEE
  Transactions on Visualization and Computer Graphics}},
  \href{https://doi.org/10.1109/TVCG.2020.3030431}{27(2):583--592},
  \href{https://doi.org/10.1109/TVCG.2020.3030431}{2021}.
  \href{https://doi.org/10.1109/TVCG.2020.3030431}
{doi: {{%
10\hspace{.1pt}\discretionary{.}{%
}{.}\hspace{.4pt}1109\discretionary{/}{%
}{/}TVCG\hspace{.1pt}\discretionary{.}{%
}{.}\hspace{.4pt}2020\hspace{.1pt}\discretionary{.}{%
}{.}\hspace{.4pt}3030431}}}


\bibitem{rotman:2013:introduction}
J.~J. Rotman.
\newblock {\em An introduction to algebraic topology}, vol. 119.
\newblock Springer Science \& Business Media, 2013.

\bibitem{Roy:19}
\href{https://doi.org/https://doi.org/10.1109/TVCG.2018.2864768}{L.~Roy,
  P.~Kumar, Y.~Zhang, and E.~Zhang}.
\newblock
  \href{https://doi.org/https://doi.org/10.1109/TVCG.2018.2864768}{Robust and
  fast extraction of 3d symmetric tensor field topology}.
\newblock \href{https://doi.org/https://doi.org/10.1109/TVCG.2018.2864768}{{\em
  {IEEE} Transactions on Visualization and Computer Graphics}},
  \href{https://doi.org/https://doi.org/10.1109/TVCG.2018.2864768}{25(1):1102--1111},
  \href{https://doi.org/https://doi.org/10.1109/TVCG.2018.2864768}{2019}.
  \href{https://doi.org/10.1109/TVCG.2018.2864768}
{doi: {{%
10\hspace{.1pt}\discretionary{.}{%
}{.}\hspace{.4pt}1109\discretionary{/}{%
}{/}TVCG\hspace{.1pt}\discretionary{.}{%
}{.}\hspace{.4pt}2018\hspace{.1pt}\discretionary{.}{%
}{.}\hspace{.4pt}2864768}}}


\bibitem{Solr-1099267}
W.~Rudin.
\newblock {\em Functional analysis}.
\newblock International series in pure and applied mathematics. McGraw-Hill,
  New York, 2nd ed. ed., 1990.

\bibitem{ABAQUS}
M.~Smith.
\newblock {\em ABAQUS/Standard User's Manual}.
\newblock Dassault Syst{\`e}mes Simulia Corp, United States, 2020.

\bibitem{tao2017semantic}
\href{https://doi.org/10.1109/TVCG.2017.2773071}{J.~Tao, C.~Wang, N.~V. Chawla,
  L.~Shi, and S.~H. Kim}.
\newblock \href{https://doi.org/10.1109/TVCG.2017.2773071}{Semantic flow graph:
  A framework for discovering object relationships in flow fields}.
\newblock \href{https://doi.org/10.1109/TVCG.2017.2773071}{{\em IEEE
  Transactions on Visualization and Computer Graphics}},
  \href{https://doi.org/10.1109/TVCG.2017.2773071}{24(12):3200--3213},
  \href{https://doi.org/10.1109/TVCG.2017.2773071}{2017}.
  \href{https://doi.org/10.1109/TVCG.2017.2773071}
{doi: {{%
10\hspace{.1pt}\discretionary{.}{%
}{.}\hspace{.4pt}1109\discretionary{/}{%
}{/}TVCG\hspace{.1pt}\discretionary{.}{%
}{.}\hspace{.4pt}2017\hspace{.1pt}\discretionary{.}{%
}{.}\hspace{.4pt}2773071}}}


\bibitem{tom:2008:algebraic}
T.~tom Dieck.
\newblock {\em Algebraic topology}, vol.~8.
\newblock European Mathematical Society, 2008.

\bibitem{Tricoche:08}
\href{http://doi.ieeecomputersociety.org/10.1109/TVCG.2008.148}{X.~Tricoche,
  G.~Kindlmann, and C.-F. Westin}.
\newblock
  \href{http://doi.ieeecomputersociety.org/10.1109/TVCG.2008.148}{Invariant
  crease lines for topological and structural analysis of tensor fields}.
\newblock \href{http://doi.ieeecomputersociety.org/10.1109/TVCG.2008.148}{{\em
  IEEE Transactions on Visualization and Computer Graphics}},
  \href{http://doi.ieeecomputersociety.org/10.1109/TVCG.2008.148}{14(6):1627--1634},
  \href{http://doi.ieeecomputersociety.org/10.1109/TVCG.2008.148}{2008}.
  \href{http://doi.ieeecomputersociety.org/10.1109/TVCG.2008.148}
{doi: {{%
10\hspace{.1pt}\discretionary{.}{%
}{.}\hspace{.4pt}1109\discretionary{/}{%
}{/}TVCG\hspace{.1pt}\discretionary{.}{%
}{.}\hspace{.4pt}2008\hspace{.1pt}\discretionary{.}{%
}{.}\hspace{.4pt}148}}}


\bibitem{Zhang:07}
\href{https://doi.org/10.1109/TVCG.2007.16}{E.~Zhang, J.~Hays, and G.~Turk}.
\newblock \href{https://doi.org/10.1109/TVCG.2007.16}{Interactive tensor field
  design and visualization on surfaces}.
\newblock \href{https://doi.org/10.1109/TVCG.2007.16}{{\em IEEE Transactions on
  Visualization and Computer Graphics}},
  \href{https://doi.org/10.1109/TVCG.2007.16}{13(1):94--107},
  \href{https://doi.org/10.1109/TVCG.2007.16}{2007}.
  \href{https://doi.org/10.1109/TVCG.2007.16}
{doi: {{%
10\hspace{.1pt}\discretionary{.}{%
}{.}\hspace{.4pt}1109\discretionary{/}{%
}{/}TVCG\hspace{.1pt}\discretionary{.}{%
}{.}\hspace{.4pt}2007\hspace{.1pt}\discretionary{.}{%
}{.}\hspace{.4pt}16}}}


\bibitem{ZhangY:17b}
\href{https://doi.org/https://doi.org/10.1007/978-3-319-61358-1_2}{Y.~Zhang,
  X.~Gao, and E.~Zhang}.
\newblock
  \href{https://doi.org/https://doi.org/10.1007/978-3-319-61358-1_2}{Applying
  2d tensor field topology to solid mechanics simulations}.
\newblock \href{https://doi.org/https://doi.org/10.1007/978-3-319-61358-1_2}{In
  T.~Schultz, E.~{\"O}zarslan, and I.~Hotz, eds., {\em Modeling, Analysis, and
  Visualization of Anisotropy}},
  \href{https://doi.org/https://doi.org/10.1007/978-3-319-61358-1_2}{pp.
  29--41}.
  \href{https://doi.org/https://doi.org/10.1007/978-3-319-61358-1_2}{Springer
  International Publishing},
  \href{https://doi.org/https://doi.org/10.1007/978-3-319-61358-1_2}{Cham},
  \href{https://doi.org/https://doi.org/10.1007/978-3-319-61358-1_2}{2017}.
  \href{https://doi.org/10.1007/978-3-319-61358-1_2}
{doi: {{%
10\hspace{.1pt}\discretionary{.}{%
}{.}\hspace{.4pt}1007\discretionary{/}{%
}{/}978\discretionary{%
}{-}{-}3\discretionary{%
}{-}{-}319\discretionary{%
}{-}{-}61358\discretionary{%
}{-}{-}1\_2}}}


\bibitem{ZhangY:20}
\href{https://doi.org/https://doi.org/10.1007/978-3-030-43036-8_15}{Y.~Zhang,
  L.~Roy, R.~Sharma, and E.~Zhang}.
\newblock
  \href{https://doi.org/https://doi.org/10.1007/978-3-030-43036-8_15}{Maximum
  number of transition points in 3d linear symmetric tensor fields}.
\newblock
  \href{https://doi.org/https://doi.org/10.1007/978-3-030-43036-8_15}{In
  H.~Carr, I.~Fujishiro, F.~Sadlo, and S.~Takahashi, eds., {\em Topological
  Methods in Data Analysis and Visualization V}},
  \href{https://doi.org/https://doi.org/10.1007/978-3-030-43036-8_15}{pp.
  237--250}.
  \href{https://doi.org/https://doi.org/10.1007/978-3-030-43036-8_15}{Springer
  International Publishing},
  \href{https://doi.org/https://doi.org/10.1007/978-3-030-43036-8_15}{Cham},
  \href{https://doi.org/https://doi.org/10.1007/978-3-030-43036-8_15}{2020}.
  \href{https://doi.org/10.1007/978-3-030-43036-8_15}
{doi: {{%
10\hspace{.1pt}\discretionary{.}{%
}{.}\hspace{.4pt}1007\discretionary{/}{%
}{/}978\discretionary{%
}{-}{-}3\discretionary{%
}{-}{-}030\discretionary{%
}{-}{-}43036\discretionary{%
}{-}{-}8\_15}}}


\bibitem{ZhangY:17a}
\href{https://doi.org/https://doi.org/10.1007/978-3-319-44684-4_13}{Y.~Zhang,
  Y.-J. Tzeng, and E.~Zhang}.
\newblock
  \href{https://doi.org/https://doi.org/10.1007/978-3-319-44684-4_13}{Maximum
  number of degenerate curves in 3d linear tensor fields}.
\newblock
  \href{https://doi.org/https://doi.org/10.1007/978-3-319-44684-4_13}{In
  H.~Carr, C.~Garth, and T.~Weinkauf, eds., {\em Topological Methods in Data
  Analysis and Visualization IV}},
  \href{https://doi.org/https://doi.org/10.1007/978-3-319-44684-4_13}{pp.
  221--234}.
  \href{https://doi.org/https://doi.org/10.1007/978-3-319-44684-4_13}{Springer
  International Publishing},
  \href{https://doi.org/https://doi.org/10.1007/978-3-319-44684-4_13}{Cham},
  \href{https://doi.org/https://doi.org/10.1007/978-3-319-44684-4_13}{2017}.
  \href{https://doi.org/10.1007/978-3-319-44684-4_13}
{doi: {{%
10\hspace{.1pt}\discretionary{.}{%
}{.}\hspace{.4pt}1007\discretionary{/}{%
}{/}978\discretionary{%
}{-}{-}3\discretionary{%
}{-}{-}319\discretionary{%
}{-}{-}44684\discretionary{%
}{-}{-}4\_13}}}


\bibitem{Zheng:04}
\href{https://doi.org/10.1109/VISUAL.2004.105}{X.~Zheng and A.~Pang}.
\newblock \href{https://doi.org/10.1109/VISUAL.2004.105}{Topological lines in
  3d tensor fields}.
\newblock \href{https://doi.org/10.1109/VISUAL.2004.105}{In {\em Proceedings
  IEEE Visualization 2004}}, \href{https://doi.org/10.1109/VISUAL.2004.105}{VIS
  '04}, \href{https://doi.org/10.1109/VISUAL.2004.105}{pp. 313--320}.
  \href{https://doi.org/10.1109/VISUAL.2004.105}{IEEE Computer Society},
  \href{https://doi.org/10.1109/VISUAL.2004.105}{Washington, DC, USA},
  \href{https://doi.org/10.1109/VISUAL.2004.105}{2004}.
  \href{https://doi.org/10.1109/VISUAL.2004.105}
{doi: {{%
10\hspace{.1pt}\discretionary{.}{%
}{.}\hspace{.4pt}1109\discretionary{/}{%
}{/}VISUAL\hspace{.1pt}\discretionary{.}{%
}{.}\hspace{.4pt}2004\hspace{.1pt}\discretionary{.}{%
}{.}\hspace{.4pt}105}}}


\bibitem{Zheng:05a}
\href{https://doi.org/10.1109/TVCG.2005.67}{X.~Zheng, B.~Parlett, and A.~Pang}.
\newblock \href{https://doi.org/10.1109/TVCG.2005.67}{Topological lines in 3d
  tensor fields and discriminant hessian factorization}.
\newblock \href{https://doi.org/10.1109/TVCG.2005.67}{{\em IEEE Transactions on
  Visualization and Computer Graphics}},
  \href{https://doi.org/10.1109/TVCG.2005.67}{11(4):395--407},
  \href{https://doi.org/10.1109/TVCG.2005.67}{2005}.
  \href{https://doi.org/10.1109/TVCG.2005.67}
{doi: {{%
10\hspace{.1pt}\discretionary{.}{%
}{.}\hspace{.4pt}1109\discretionary{/}{%
}{/}TVCG\hspace{.1pt}\discretionary{.}{%
}{.}\hspace{.4pt}2005\hspace{.1pt}\discretionary{.}{%
}{.}\hspace{.4pt}67}}}


\bibitem{Zheng:05b}
\href{https://doi.org/10.1109/VISUAL.2005.1532841}{X.~Zheng, B.~Parlett, and
  A.~Pang}.
\newblock \href{https://doi.org/10.1109/VISUAL.2005.1532841}{{Topological
  structures of 3D tensor fields}}.
\newblock \href{https://doi.org/10.1109/VISUAL.2005.1532841}{In {\em
  Proceedings IEEE Visualization 2005}},
  \href{https://doi.org/10.1109/VISUAL.2005.1532841}{pp. 551--558},
  \href{https://doi.org/10.1109/VISUAL.2005.1532841}{2005}.
  \href{https://doi.org/10.1109/VISUAL.2005.1532841}
{doi: {{%
10\hspace{.1pt}\discretionary{.}{%
}{.}\hspace{.4pt}1109\discretionary{/}{%
}{/}VISUAL\hspace{.1pt}\discretionary{.}{%
}{.}\hspace{.4pt}2005\hspace{.1pt}\discretionary{.}{%
}{.}\hspace{.4pt}1532841}}}


\end{thebibliography}

\clearpage
\appendix

\section{Degenerate Point Index}

In this section, we provide detail of our theoretical analysis on the index of degenerate points in a 3D symmetric tensor field (Section~\ref{sec:index}).

Let $R$ be a topological disk without self-intersections such that there are no degenerate points on its boundary, $\partial R$ (the circles in Figure~\ref{fig:thm_2} (a-b)). We consider the right-handed frames formed by the unit major eigenvector $v_1$, the medium eigenvector $v_2$, and the minor eigenvector $v_3$ of the tensor fields on $\partial R$. There are four ways of selecting a right-handed frame from the eigenvectors. Let $f_0(p)=(v_1, v_2, v_3)$ be one such frame. Then $f_1(p)=(v_1, -v_2, -v_3)$, $f_2(p)=(-v_1, v_2, -v_3)$, and $f_3(p)=(-v_1, -v_2, v_3)$ are the other choices of such frames (Figure~\ref{fig:thm_2} (left)). Let $r_m$ ($0\le m \le 3$) be the 3D rotation that maps the $X$-axis to the major eigenvector in $f_m(p)$, the $Y$-axis to the medium eigenvector, and the $Z$-axis to the minor eigenvector (Figure~\ref{fig:thm_2} (left)). Define $r_x$, $r_y$, and $r_z$ as the $180^\circ$ rotation around the $X$-, $Y$-, and $Z$-axis, respectively. Then we have

\begin{eqnarray}
  r_1 &=& r_0 r_x \\
  r_2 &=& r_0 r_y \\
  r_3 &=& r_0 r_z
\end{eqnarray}

We choose $p_0 \in \partial R$ and travel along $\partial R$ for one round in order to inspect the behavior of the {\em continuous} eigenframe that is initially set to be $f_0(p_0)$ (Figure~\ref{fig:thm_2} (a-b)). Since the tensor field is continuous over $R$ and there is no degenerate point on $\partial R$, we know that the eigenvector fields are also continuous over $\partial R$. Therefore, when returning to $p_0$ after a full boundary walk, the frame $f'(p_0)$ must be $f_m(p_0)$ for some $0\le m\le 3$. That is,

\begin{equation}
  r'(p_0)=r_0(p_0) c
\end{equation}

\noindent where $c=1$, $r_x$, $r_y$, or $r_z$. Thus $c=r_0(p_0)^{-1}r'(p_0)$. Next, we show that $c$ is a curve invariant.

\begin{lemma}
Given the conditions above, $c$ is independent of the coordinate system for the space.
\label{lemma:free_coordinate_system}
\end{lemma}

\begin{proof}
Let $S$ and $T$ be two right-handed orthogonal coordinate systems (Figure~\ref{fig:lemma_1}). Let $P$ be the change-of-basis matrix from $T$ to $S$. Note that $P^{-1}$ can also be considered as the rotation that takes the coordinate system $S$ to the coordinate system $T$.

Therefore, $r_{0,S}(p_0)$, the rotation that takes $S$ to the eigenframe $f_0(p_0)$, is related to $r_{0,T}(p_0)$, the rotation that takes the $T$ to the eigenframe $f_0(p_0)$, as follows: $r_{0,S}(p_0) = r_{0, T}(p_0) P^{-1}$. Similarly, $r'_{S}(p_0) = r'_{T}(p_0) P^{-1}$.

We now consider $c_S=(r_{0,S}(p_0))^{-1} r'_{S}(p_0)$. This is equivalent to $c_S = (r_{0,T}(p_0) P^{-1})^{-1} r'_{T}(p_0) P^{-1} = P (r_{0, T}(p_0))^{-1} r'_{T}(p_0) P^{-1} = P c_T|_S P^{-1}$.

We consider the matrices corresponding to $c_S$ and $c_T$ under $S$. Thus, $c_S|_S = P|_S (c_T|_S) P^{-1}|_S$. Since $P|_S$ and $P|_{S}^{-1}$ commute, we have $P|_S=P|_T$ and $P^{-1}|_S = P^{-1}|_T$. Thus, we drop the subscripts and only use $P$ and $P^{-1}$. On the other hand, $c_T|_S = P^{-1}(c_T|_T) P$. Thus, $c_S|_S = P (c_T|_S) P^{-1} = P (P^{-1}(c_T|_T) P) P^{-1} = c_T|_T$.

Therefore, $c$ is independent of the choice of the coordinate system.
\end{proof}

\begin{lemma}
Given the conditions above, $c$ is independent of the initial frame chosen.
\label{lemma:free_initial_frame}
\end{lemma}

\begin{proof}
Assume that we have chosen $f_m(p_0)$ as the initial frame where $m \ne 0$. Let $p_1$, ... $p_n$ be a sequence of points on $\partial R$ in the direction of the travel such that $p_n=p_0$ (Figure~\ref{fig:thm_2} (a-b)). Moreover, we choose the sample points so that at each sample point $p_\ell$, there is a unique $k_0$ that minimizes $d(f_0(p_\ell), f_k(p_{\ell+1}))$, the distance between the two frames (the angle of the minimal 3D rotation that takes the first frame to the second frame). For convenience, we reorder the four frames at $p_{\ell+1}$ such that $f_{k_0}(p_{\ell+1})$ becomes $f_0(p_{\ell+1})$. Under the new numberings, we define $s_{\ell+1}$ to be the unique rotation that takes the frame $f_0(p_{\ell})$ to $f_0(p_{\ell+1})$. That is, $s_\ell = r_0(p_{\ell+1}) (r_0(p_{\ell}))^{-1}$. Note that $s_\ell (f_m(p_{\ell}))=f_0(p_{\ell+1})$ for $1\le m \le 3$. Therefore,

\begin{equation}
r_m'(p_0)  = s_n s_{n-1} .. s_1  r_m(p_0) = s_n s_{n-1} .. s_1 r_0(p_0) c = r'(p_0) c
\end{equation}

Consequently, $c'=(r_m(p_0))^{-1}r_m'(p_0) = c^{-1}(r_0(p_0))^{-1} r'(p_0) c = c^{-1} c c = c$. Thus, $c$ does not depend on the choice of the initial frame.
\end{proof}

\begin{lemma}
Given the conditions above, $c$ is independent of the direction of travel along $\partial R$.
\label{lemma:free_orientation}
\end{lemma}

\begin{proof}
Assume that we have chosen $f_0(p_0)$ as the initial frame. Let $p_1$, ... $p_n$ be the sequence of points on $\partial R$ from Lemma~\ref{lemma:free_initial_frame} (Figure~\ref{fig:thm_2} (a-b)). Furthermore, let $s_\ell = r_0(p_{\ell+1}) (r_0(p_{\ell}))^{-1}$ be the unique rotation that takes the frame $f_0(p_{\ell})$ to $f_0(p_{\ell+1})$. Then we have $s_\ell^{-1} (f_m(p_{\ell+1}))=f_0(p_{\ell})$ for $1\le \ell \le n$ and $0\le m \le 3$.

We now consider travelling in the opposite direction along $\partial R$, that is, $p_0$, $p_1'=p_{n-1}$, ... $p_{n-1}'=p_1$ and $p_n'=p_0$. Therefore, $c' = (r_0(p_0))^{-1} s_n^{-1} ... s_1^{-1} r_0(p_0) = ((r_0(p_0))^{-1} s_1 ... s_n r_0(p_0))^{-1} = (r_0(p_0))^{-1} r'(p_0))^{-1}=c^{-1}=c$.

Thus, $c$ does not depend on the travel direction.
\end{proof}

\begin{figure}[!t]
	\centering
	\begin{overpic}[width={0.6\columnwidth}]{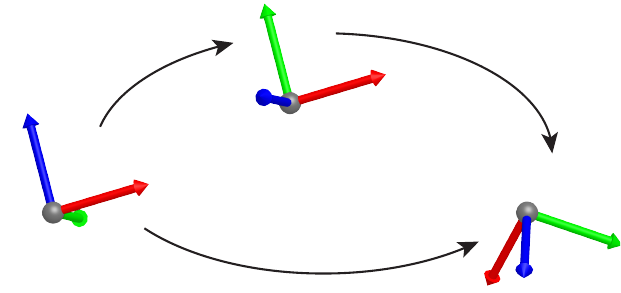}
		\put(140, 20){$f_0(p_0)$}
		\put(30, 60){$P^{-1}$}
		\put(70, 65){$S$}
		\put(130, 50){$r_{0,S}$}
		\put(20, 30){$T$}
		\put(60, 10){$r_{0,T}$}
	\end{overpic}
	\caption{This figure illustrates the change of the basis of $S$, $T$, and $f_0(p_0)$.
	}
	\label{fig:lemma_1}
\end{figure}

\begin{lemma}
	Given the conditions above, $c$ is independent of the choice of the starting point $p_0$.
	\label{lemma:free_starting_point}
\end{lemma}

\begin{proof}
Let $p_0 \ne p_0'$ be two points on $\partial R$. Choose $f_0(p_0)$ as the initial frame. Let $p_1$, ... $p_n$ be a sequence of points on $\partial R$ described in Lemma~\ref{lemma:free_initial_frame} (Figure~\ref{fig:thm_2} (a-b)). Moreover, assume that $p_k=p_0'$ for some $1 \le k < n$. Using the process described Lemma~\ref{lemma:free_initial_frame}, we can find $f_0(p_0')$ such that $r_0(p_0')=s_{k} s_{k-1}...s_1 r_0(p_0)$.

We now consider travelling starting from $p_0'$ in the sequence of $p_{k+1}, .. p_{n-1}, p_0, ... p_{k-1}, p_k=p_0'$. Then, $r'(p_0') = s_k s_{k-1} ... s_1 s_n...s_{k+2} s_{k+1}r_0(p_0')$. Therefore, $c'=(r_0(p_0'))^{-1} r'(p_0') = (s_{k} s_{k-1}...s_1 r_0(p_0'))^{-1} (s_k s_{k-1} ... s_1 s_n...s_{k+2} s_{k+1}r_0(p_0')) = (r_0(p_0))^{-1} s_n ... s_{k+1} r_0(p_0') = (r_0(p_0))^{-1} s_n ... s_{k+1} (s_k .. s_1 r_0(p_0)) = c$.

Therefore, $c$ does not depend on the choice of the initial point $p_0$.
\end{proof}

\begin{lemma}
Given two topological disks $R_1$ and $R_2$ that intersect only at their common boundary such that $R_1\bigcup R_2$ is still a topological disk, the winding number of the boundary of $R_1 \bigcup R_2$ is the product of the winding numbers of $R_1$ and $R_2$.
\label{lemma:product}
\end{lemma}

\begin{proof}
Note that the boundary $\partial (R_1\bigcup R_2)=\partial R_1 \bigcup \partial R_2$. Let $p_0 \in \partial R_1 \bigcap \partial R_2$ be the starting point and travel the boundary of $R_1\bigcup R_2$. This is equivalent to traveling around the boundary of $R_1$ and $R_2$ once each, passing through $p_0$ once before returning to it a second time. Since the quaternion of a 3D rotation from concatenating two 3D rotations is the product of the quaternions for the two 3D rotations, the winding number of the boundary of $R_1 \bigcup R_2$ is the product of the respective winding numbers of the boundary of $R_1$ and the boundary of $R_2$.
\end{proof}

\begin{lemma}
Given a 3D tensor field $T(x, y, z)$ and $R$, a topological ball on which $T(x, y, z)$ is linear, the major eigenvector field of $T(x, y, z)$, a 3D line field, can be turned into a 3D vector field inside $R$. Similarly, if $T(x, y, z)$ is planar on $R$, then the minor eigenvector field of $T(x, y, z)$ can be turned into a 3D vector field.
\label{lemma:line_field_to_vector_field}
\end{lemma}

\begin{proof}
Since $R$ is a topological ball, it is simply connected and finite. If $T(x, y, z)$ is linear on $R$, the major eigenvector field is always well-defined in $R$, i.e., without singularities. Thus, it can be turned into a 3D vector field from the result of Markus~\cite{Markus:55}, which states that a 3D line field can be turned into a 3D vector field on a simply-connected, finite region if the line field does not have any singularities in the region. The proof for the minor eigenvector field in planar regions is similar.
\end{proof}

\begin{theorem}
Given a 3D tensor field $T(x, y, z)$ and a point $p_0$ (possibly a degenerate point), there exists a small enough neighborhood $R$ of $p_0$ such that any topological disk inside $R$ has the same winding number for its boundary if the disk contains $p_0$ but no other degenerate points of $T(x, y, z)$ and has no self-intersection. In this case, the winding number is $\mathbf{i}$ if $p_0$ is a linear wedge, $-\mathbf{i}$ if $p_0$ is a linear trisector, $\mathbf{k}$ if $p_0$ is planar wedge, $\mathbf{-k}$ if $p_0$ is planar trisector, and $\mathbf{1}$ if $p_0$ is not a degenerate point.
\label{thm:winding_number_local}
\end{theorem}

\begin{proof}
If $p_0$ is a linear degenerate point, then there exists a neighborhood $R$ of $p_0$ such that the major eigenvector field of $T(x, y, z)$ can be converted to a continuous vector field inside $R$ without singularities (Lemma~\ref{lemma:line_field_to_vector_field}). According to the Flow Box theorem~\cite{Calcaterra:2008}, there exists a region $R'\subset R$ and a diffeomorphism $\phi$ from $R'$ to another space $F$ such that the major eigenvector field (now a vector field on $R'$) to a constant vector field defined on $F$. Without loss of generality, we can assume that the major eigenvector at $p_0$ has the same length as its image under $\phi$, \ie a 3D rotation without scaling. Since $\phi$ is a diffeomorphism, it is continuous. Therefore, it is possible to find an even smaller set $R'' \subset R'$ such that the diffeomorphism $\phi$ inside $R''$ can be approximated by $\phi(p_0)$ with a sufficiently small error $\epsilon.$ Notice that under $\phi$, the tensor field $T' = \phi(T(x, y, z))$ is also a tensor field whose major eigenvectors are all parallel. In addition, $\phi(p_0)$ is a linear degenerate point of $T'$. Similarly, a topological disk $D$ containing $p_0$ will be mapped to a topological disk containing $\phi(p_0)$. When $\epsilon$ is small enough, the winding number of the boundary of $D$ is the same as the winding number of the boundary of $\phi(D)$ since $\phi(p_0)$ is a 3D rotation on the eigenvectors of the tensor field. We can further select $R''$ to be small enough such that any loop inside $R''$ is close to being planar, i.e., contained in some plane. Consequently, the image of such a loop under $\phi$ is also nearly a planar loop. We select a point $p_0$ and travel along the loop. Based on Lemma~\ref{lemma:free_coordinate_system}, we can choose any coordinate system and the winding number will not change. Thus, for simplicity, we choose the eigenframe at the start point $p_0$ to be coordinate system. Therefore, the quaternion for $p_0$ is $\mathbf{1}$. Since the major eigenvector field is constant along the loop, the quaternions corresponding to the eigenframes along the loop have the form $w+x\mathbf{i}$, i.e., no $\mathbf{j}$ and $\mathbf{k}$ components. Thus, when returning to $p_0$, the quaternion corresponding to $p_0$ must be $\pm \mathbf{1}$ or $\pm \mathbf{i}$. When the region $R$ contains no singularity, the winding number is $\mathbf{1}$. Otherwise, it is $\mathbf{i}$ if the singularity contained in $R$ is a wedge or $-\mathbf{i}$ if the singularity is a trisector.

Similarly, in a planar region, the winding number is $\mathbf{1}$ if the region contains no singularity, and is $\mathbf{k}$ or $-\mathbf{k}$ if the singularity is a wedge or a trisector, respectively.
\end{proof}

\begin{corollary}
Given a 3D tensor field $T(x, y, z)$ and a topological disk $R$ free of self-intersections, assume that $R$ contains only one degenerate point inside. If furthermore the normal to the surface $R$ is nowhere perpendicular to the dominant eigenvector field, then the winding number of the boundary $\partial R$ is the same as the index of the degenerate point.
\label{coro:global}
\end{corollary}

\begin{proof}
Given any point $p$ in $R$, there exists a sufficiently small neighborhood $U_p$ such that Theorem~\ref{thm:winding_number_local} is satisfied. These neighborhoods give an open cover of $R$. Since $R$ is finite and closed, any of its open covers has a finite subcover~\cite{Solr-1099267}. Consequently, we can find a finite neighborhood $U_1$, $U_2$, ... $U_m$ for some $m>0$ such that their union covers $R$. In Addition, $\partial R$, a loop, is covered by $U_1$, ... $U_m$. It is thus possible to decompose $\partial R$ as the union of a number of closed curves, each of which is inside one such neighborhood $U_k$ for some $1\le k\le m$. Therefore, the winding number of $\partial R$ is the product of the winding numbers of each of such closed curves (Lemma~\ref{lemma:product}). Since $R$ is nowhere perpendicular to the dominant eigenvector field, the dot product between surface normals (chosen consistently over $R$) and the dominant eigenvector field over $R$ is either always positive or always negative.

Since $R$ contains only one degenerate point, we can select the closed curves in the open cover such that the degenerate point is inside only one topological disk bounded by the closed curves. For this curve, the winding number is either $\mathbf{i}$, $-\mathbf{i}$, $\mathbf{k}$, or $-\mathbf{k}$ while for the other closed curves, the winding number is $\mathbf{1}$. Thus, the winding number of $\partial R$ is $\pm \mathbf{i}$ or $\pm \mathbf{k}$ due to Lemma~\ref{lemma:product}.
\end{proof}

\section{Jones Polynomial Computation}
\label{sec:jones}

In this section, we provide some detail on the technique of computing the Jones polynomials~\cite{livingston:1993:knot}, which we implement in our system. Recall that Jones polynomial of a given curve network is defined in terms of the link diagram of curve network, though it is an invariant as it does not depend on the actual choice of the plane onto which the curve network is projected.

As the Jones polynomial is defined recursively, computing it is an NP-hard problem in terms of the number of crossing points in the link diagram. Livingston~\cite{livingston:1993:knot} provides an approximation algorithm by involving the {\em Reidemeister moves} (Figure~\ref{fig:Knot_Reidemeister}), which can reduce the number of crossing points without changing the polynomial itself. Note that the reconnection operations in the definition of the Jones polynomials (Figure~\ref{fig:jones} (a-2.1 and a-2.2)) of a curve network generates two new curve networks, each of which has one fewer crossing points than the original curve network. However, the Jones polynomials of the two new networks are usually different from that of the original one, hence the exponential growth of the computation time in terms of the number of crossing points in the original curve network.

\begin{figure}[!t]
	\centering
	\begin{overpic}[width={0.9\columnwidth}]{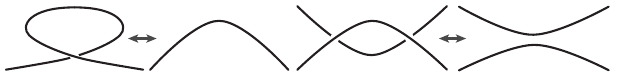}
		\put(50, -8)  {\small Type I}
		\put(160, -8)  {\small Type II}
	\end{overpic}
	\caption{The two types of Reidemeister moves that can reduce the number of crossing points in a link diagram.
	}
	\label{fig:Knot_Reidemeister}
\end{figure}

We follow closely the technique of Livingston~\cite{livingston:1993:knot}. Given a degenerate loop in the 3D space, we apply the principal component analysis on the points on the loop, which gives us a new coordinate system. We next project the curve onto the $XY$-plane of the new coordinate system and ensure we have a regular projection where there are no overlapping edges and no three points projected to the same point on the plane~\cite{livingston:1993:knot}. If it is not a regular projection, we apply a small but random 3D rotation to the coordinate system from the principal component analysis. Since irregular projections are structurally unstable, an arbitrarily small perturbation in the coordination system can usually generate a regular projection. We then construct the link diagram by tracing the crossings on the projected curve. Since the complexity of the Jones polynomial computation is $O(2^n)$ for $n$ crossing points, we simplify the link diagram by reducing the number of crossings points with types I and II Reidemeister moves. This involves the computation of {\em braid groups}, and we refer our readers to~\cite{livingston:1993:knot} for more detail on this part of the algorithm. Lastly, we compute the Jones polynomial of the simplified curve network by applying the first and second type of simplification rules~\cite{livingston:1993:knot}. We iteratively perform the following two steps. First, any loop in the network free of crossing points with the rest of the networks is removed. When no such loop exists in the remaining network, we remove a crossing with two local reconnection. This leads to two new networks, which are sent to the same routine to compute their respective Jones polynomial. This recursion will eventually lead to the Jones polynomial of the original curve network. Note that we only compute the Jones polynomial for individual degenerate loops for knot identification and classification, even though the computation of the Jones polynomial of a loop may result in a curve network in the middle of the computation due to the simplification rules. Finally, when the Jones polynomial is not a constant, we regard the degenerate loop as a knot and add $\ast$ to the corresponding node in the topological graph.

\section{Metal Hook}
\begin{wrapfigure}{r}{0.25\columnwidth}
	\vspace{-2em}
	\begin{center}
		\includegraphics[width=0.25\columnwidth]{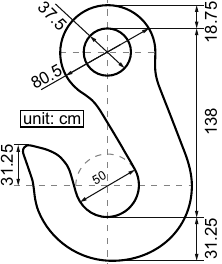}
	\end{center}
	\vspace{-2em}
\end{wrapfigure}

Metal hooks are essential to the lifting mechanism used in construction and transportation engineering.  They are often used to lift or pull different kinds of goods that are of varying sizes and shapes.  To be able to create cost-effective designs for these hooks, it is important to evaluate their use under different loading conditions.  Here we contrast two types of loading conditions for the metal hook example published in SIMULIA $2020$.  The first type loads the hook from one angle that is shifted to the left from the vertical direction with $1$ kN and the second type adds a load pulling in the horizontal direction with $1$ kN.  The latter simulates a use where the hook carries more than one load.  We note that the two scenarios produce very different numbers of regions and degenerate curves.   The Betti numbers for the single load scenario are higher for the planar and the linear regions.  This was not expected due to the simple load and could not be observed without the topological graph.  Moreover, for the second scenario with two loads, there is no region containment as observed for the first scenario in which two planar regions each contains a linear region.  This comparison inspires more testing to corroborate the general practice where multiple loads are sometimes applied to ensure little rotational movement of the hook to avoid swinging or tipping.
\begin{figure}[!t]
	\centering
	\begin{overpic}[width={\columnwidth}]{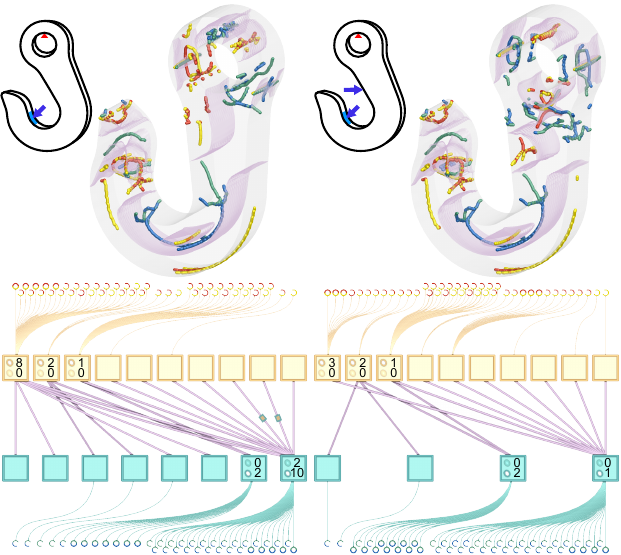}
		\put(50, -8)  {\small(a)}
		\put(170, -8)  {\small(b)}
	\end{overpic}
	\caption{ This figure shows the simulated metal hook data with either a single-load scenario (a) or a multi-load scenario (b).
	}
	\label{fig:hook}
\end{figure}

\section{Eight Pile Group Foundation with Cap}
\begin{figure}[!t]
	\centering
	\includegraphics[width=0.9\columnwidth]{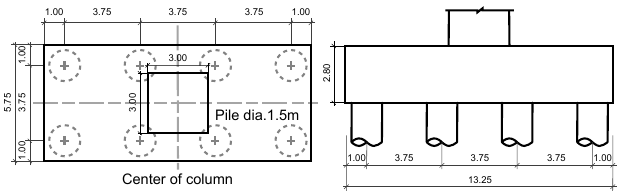}
	\caption{In this eight pile group, there is a cap that is $13.25 \times 5.75 \times 2.80$ volume meters on top of the piles. Each pile has a diameter of $1.5$ meters.  The cap is attached to the piles with no movement allowed. The ends of the piles are fixed.
	}
	\label{fig:pile8_setting}
\end{figure}

\begin{figure}[!t]
	\centering
	\begin{overpic}[width={\columnwidth}]{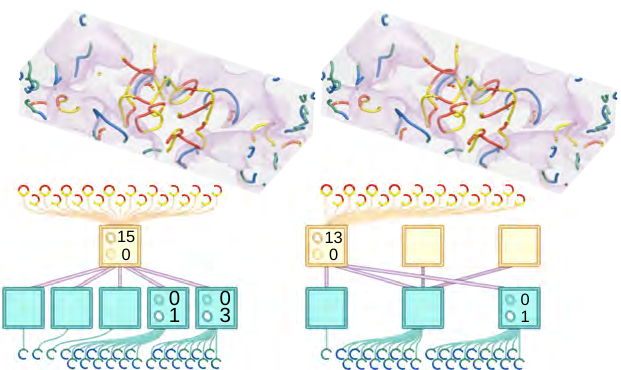}
		\put(20,   -7) {\small(a) Poisson's ratio $0.13$}
		\put(130, -7)  {\small(b) Poisson's ratio $0.24$}
	\end{overpic}
	\caption{Tensor field topology and its topological graph for the pile cap in an eight pile group foundation.   The first case (left) has its four center piles at a lower Poisson's ratio ($0.13$) and its four piles at the ends at a higher Poisson's ratio ($0.20$).  The second case (right) has its four center piles at a higher Poisson's ratio ($0.24$) while maintaining the Poisson's ratios of the four piles at the ends at $0.20$.
	}
	\label{fig:pile8_1}
\end{figure}

Pile groups are effective foundation structures that support buildings or bridges~\cite{chaimahawan:2021:finite}.  Figure~\ref{fig:pile8_setting} illustrates an eight pile group. Knowing how the load distributes and how the material deforms is important to structure integrity evaluation and maintenance scheduling.  The most common material used for these pile groups is concrete which is of crushed stones, sand and water. The mixing causes the concrete to have different Poisson's ratios; that is, a pile group may have a range of Poisson's ratios for its piles.  Here we contrast two cases where the variation of Poisson's ratios is different. The first case has $0.13$ for the center $4$ piles and $0.20$ for the $4$ piles on the ends, while the second increases the Poisson's ratios for the middle $4$ piles to $0.24$. The latter simulates a more incompressible center than the ends and the former is the reverse.  We add a vertical load, $1000$ kN, at the center top ($3 \times 3$ square meters) and a small periodic loading to the sides of the pile cap.

\begin{wrapfigure}{r}{0.5\columnwidth}
	\vspace{-2em}
	\begin{center}
		\includegraphics[width=0.5\columnwidth]{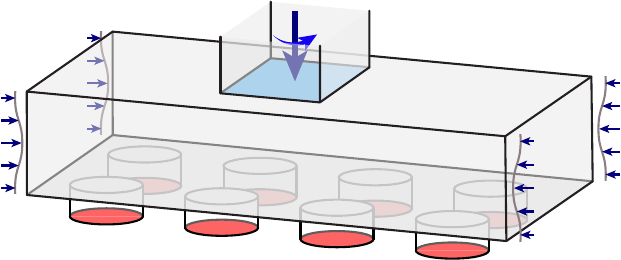}
	\end{center}
	\vspace{-2em}
\end{wrapfigure}

Using the topological graphs (Figure~\ref{fig:pile8_1}), we can observe that there are more linear regions representing extensions than the planar regions representing compression for the first case while there are the same number of linear and planar regions in the second case.  The Betti numbers are higher for the regions in the first case.   Conjectures such as having a more incompressible center group, i.e., higher Poisson's ratio for the center piles, leads to a more uniform material behavior for the cap become plausible; however, only extensive studies can warrant these statements. Our topological graphs can aid in establishing these conjectures to provide practical guidance on the pile arrangement for long-lasting concrete foundations.

\begin{figure*}[!t]
	\centering
	\begin{overpic}[width={0.9\textwidth}]{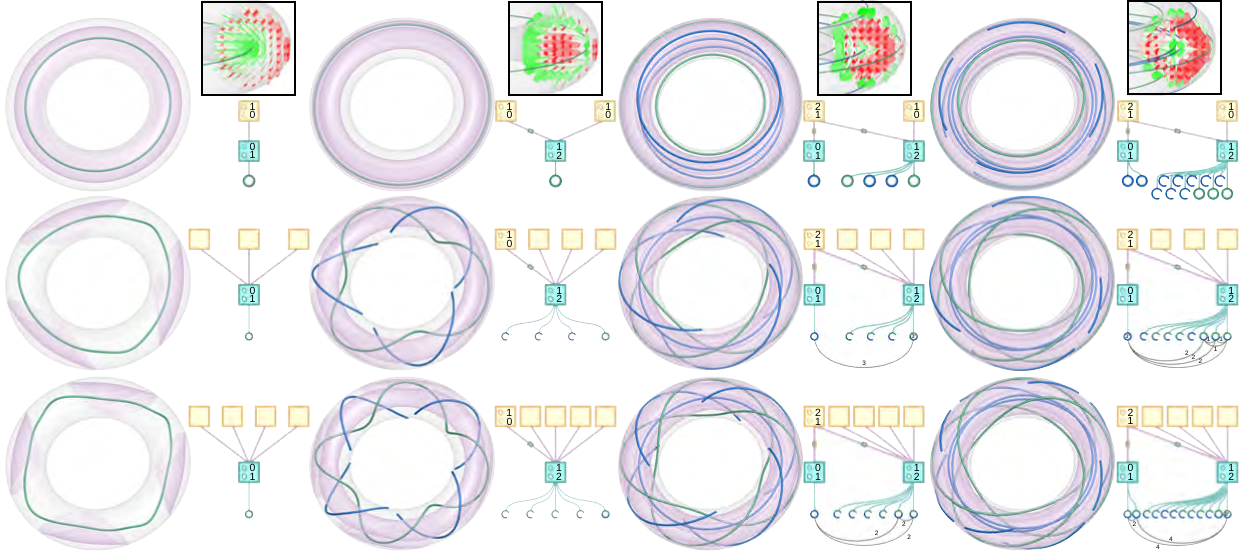}
		\put(  30, -7)  {\small$q = 1$}
		\put(150, -7)  {\small$q = 2$}
		\put(260, -7)  {\small$q = 3$}
		\put(375, -7)  {\small$q = 4$}
		
		\put(-20, 165)  {\small$p = 0$}
		\put(-20, 100)  {\small$p = 3$}
		\put(-20, 25)  {\small$p = 4$}
	\end{overpic}
	\caption{We evaluate the compression force of  $p = 0, 3, 4$, $q = 1, 3, 4$, and $\alpha=0.3$. For each scenario, we show the tensor field features on the left and the topological graph on the right. As $p$ increases, the number of the features such as degenerate curves increases due to the periodicity of the compression force. On the other hand, As $q$ increases, trisector degenerate curves start to appear. We also show the glyphs on a cross-section for the case of $p=0$ to indicate the eigenvector fields.
	}
	\label{fig:Oring_9}
\end{figure*}

\section{Additional O-ring Analysis}
\label{sec:oring_more}

In this section, we add some discussions on the results from varying the periodicity of loading on the O-ring.  We provide a group of examples while varying $p$ and $q$ in Equation~\ref{eq:traction_force}.

Our results are shown in Figure~\ref{fig:Oring_9} as an array of sub-figures.  Along the vertical direction, the value of $p$ increases and along the horizontal direction, $q$ increases.  We first note that for high values of $p$ and $q$, the topological graphs are more complex in terms of increased number of regions and degenerate curves.  For $q=1$ which is shown in the leftmost column, there is no region containment.  For $q=2$, a new planar region appears for each value of $p$, and it resides inside the green region.  One of the straight purple edges representing the neutral surfaces is encoded with the glyph that shows containment.  For $q=3$, this planar region starts to host a linear region, which itself contains another planar region.  The appearance of the nested regions is a direct response of material deformation to the boundary condition change.  Furthermore, linking among linear degenerate curves starts to form.  In particular, for $p=q=3$, a trefoil appears.  As $q$ increases to $4$, this trefoil breaks and many more degenerate curves appear and link.   For $p=q=4$, a knot with a Writhe number $8$ appears.  More complex linking and knotting behaviors are expected as $p$ and $q$ increase.  During this deformation, degenerate curves that are either extension or compression dominant will intertwine with each other.  The knottiness may indicate a match of the physical domain boundary with the high stress loading spots that enclose the same material deformation behavior entirely inside the physical domain.  In future work, we plan to continue with the experiments for many more scenarios to detect patterns and mathematical reasons.

\end{document}